\documentclass[oneside,english]{amsart}
\usepackage[T1]{fontenc}
\usepackage[latin9]{inputenc}
\usepackage{babel}
\usepackage{textcomp}
\usepackage{amstext}
\usepackage{amsthm}
\usepackage{amssymb}
\usepackage{graphicx}
\usepackage[unicode=true,
 bookmarks=false,
 breaklinks=false,pdfborder={0 0 1},backref=section,colorlinks=false]
 {hyperref}

\makeatletter

\providecommand{\tabularnewline}{\\}

\numberwithin{equation}{section}
\numberwithin{figure}{section}
\theoremstyle{plain}
\newtheorem{thm}{\protect\theoremname}
\theoremstyle{remark}
\newtheorem{rem}[thm]{\protect\remarkname}
\theoremstyle{plain}
\newtheorem{lem}[thm]{\protect\lemmaname}
\theoremstyle{plain}
\newtheorem{prop}[thm]{\protect\propositionname}

\usepackage{babel}

\makeatother

\providecommand{\lemmaname}{Lemma}
\providecommand{\propositionname}{Proposition}
\providecommand{\remarkname}{Remark}
\providecommand{\theoremname}{Theorem}

\begin{document}
\title{$F$-theory over a Fano threefold built from $A_{4}$-roots }
\author{Herbert Clemens and Stuart Raby}
\date{December 19, 2021}
\email{clemens.43@osu.edu, raby.1@osu.edu}
\begin{abstract}
In previous papers, the authors showed the advantages of building
a $\mathbb{Z}_{2}$-action into an $F$-theory model $W_{4}/B_{3}$,
namely the action of complex conjugation on the complex algebraic
group with compact real form $E_{8}$. The goal of this paper is to
construct the Fano threefold $B_{3}$ directly from the roots of $SU\left(5\right)$
in such a way that the action of complex conjugation is exactly the
desired $\mathbb{Z}_{2}$-action and the quotient of this action on
$W_{4}/B_{3}$ and its Heterotic dual have the phenomenologically
correct invariants. 
\end{abstract}

\maketitle
\tableofcontents{}

\section{Introduction}

A particular challenge in Heterotic $F$-theory duality arises when
one wishes to transfer a $\mathbb{Z}_{2}$-action 
\begin{equation}
V_{3}^{\vee}/B_{2}^{\vee}=\frac{V_{3}/B_{2}}{\mathbb{Z}_{2}}\label{eq:Het}
\end{equation}
on a elliptically fibered Calabi-Yau Heterotic threefold $V_{3}/B_{2}$
to a $\mathbb{Z}_{2}$-action 
\begin{equation}
W_{4}^{\vee}/B_{3}^{\vee}=\frac{W_{4}/B_{3}}{\mathbb{Z}_{2}}\label{eq:F-th}
\end{equation}
on an elliptically fibered Calabi-Yau fourfold that becomes the $F$-theory
dual. We have proposed a framework for such a duality in \cite{Clemens-1}.
$W_{4}/B_{3}$ with $B_{3}=B_{2}\times\mathbb{P}_{\left[u_{0},v_{0}\right]}$
with del Pezzo $B_{2}$ is defined by a \textit{Tate form} 
\begin{equation}
wy^{2}=x^{3}+a_{5}xyw+a_{4}zx^{2}w+a_{3}z^{2}yw^{2}+a_{2}z^{3}xw^{2}+a_{0}z^{5}w^{3}\label{eq:Tate}
\end{equation}
with $a_{j},z,\frac{y}{x}\in H^{0}\left(K_{B_{3}}^{-1}\right)^{\left[-1\right]}$
with respect to the $\mathbb{Z}_{2}$-action. We will require that
$W_{4}/B_{3}$ be defined subject to the condition 
\begin{equation}
a_{5}+a_{4}+a_{3}+a_{2}+a_{0}=0.\label{eq:seccond}
\end{equation}
The condition is equivalent to the condition that $W_{4}/B_{3}$ have
a second section $\tau$ given by 
\[
\begin{array}{c}
x=wz^{2}\\
y=wz^{3}.
\end{array}
\]
Incorporating translation by the difference of $\tau$ and the standard
section $\zeta$ given by 
\[
\begin{array}{c}
x=0\\
w=0
\end{array}
\]
into the $\mathbb{Z}_{2}$-action allows us to eliminate vector-like
exotics in a final paper \cite{Clemens-3} of this series.

Furthermore $B_{3}$ is a $\mathbb{P}^{1}$-fiber bundle over $B_{2}$
on which $\mathbb{Z}_{2}$ must act equivariantly. One desires such
a configuration in order to employ the Wilson line mechanism for symmetry-breaking
consistently and simultaneously on both the Heterotic model and its
$F$-theory dual.

The $\mathbb{Z}_{2}$-action on a $V_{3}$ must be free. Furthermore
on the $F$-theory side the $\mathbb{Z}_{2}$-action must restrict
to a free $\mathbb{Z}_{2}$-action on a distinguished smooth anti-canonical
divisor $S_{\mathrm{GUT}}\subseteq B_{3}$. Therefore it must act
skew-symmetrically on the anti-canonical section $z$ defining $S_{\mathrm{GUT}}$.
In \cite{Clemens-1} we showed that, while $x$ and $w$ are symmetric
with respect to the $\mathbb{Z}_{2}$-action on (\ref{eq:Tate}),
$y$, $z$, and all the $a_{j}$ have to be skew-symmetric.

The Heterotic $\mathbb{Z}_{2}$-action must preserve the initial $E_{8}$-symmetry
and so the $F$-theory $\mathbb{Z}_{2}$-action must preserve initial
$E_{8}$-symmetry as well. In short, the challenge is to begin with
$E_{8}$-symmetry on both the Heterotic and $F$-theory sides and,
for successive subgroups $G_{\mathbb{R}}\leq E_{8},$ to match breaking
to $G_{\mathbb{R}}$-symmetry on the Heterotic side with simultaneous
breaking to $G_{\mathbb{R}}$-symmetry on the $F$-theory side throughout,
ending with symmetry-breaking to $G_{\mathbb{R}}=SU\left(3\right)\times SU\left(2\right)\times U\left(1\right)$,
the so-called Minimal Supersymmetric Standard Model {[}MSSM{]}.

As we showed in \cite{Clemens-1} , the necessity that $\mathbb{Z}_{2}$
must act as 
\[
\frac{dx}{y}\mapsto-\frac{dx}{y}
\]
on the relative one-form of the elliptic fibration $W_{4}/B_{3}$
implies that it must incorporate the central involution 
\[
-I_{8}:\mathfrak{h}_{E_{8}^{\mathbb{C}}}\rightarrow\mathfrak{h}_{E_{8}^{\mathbb{C}}}
\]
on the Cartan subalgebra of the complex algebraic group $E_{8}^{\mathbb{C}}$
at the outset \textit{without} breaking initial $E_{8}$-symmetry
on the quotients \eqref{eq:Het})and \eqref{eq:F-th}.

To achieve this, in \cite{Clemens-1} we proposed the method of replacing
all roots $\rho$ with $-\rho$ via the operation of complex conjugation
on the complex algebraic group $E_{8}^{\mathbb{C}}$ and all relevant
subgroups $G_{\mathbb{C}}$, an operation that restricts to the identity
on all compact real forms $G_{\mathbb{R}}$. Since all the compact
real forms have faithful real matrix representations, this complex
conjugation operator will not affect $G_{\mathbb{R}}$-symmetry and
will commute with the various symmetry-breaking steps.

The purpose of this paper is to build the appropriate base space $B_{3}$
of the elliptically fibered $F$-theory model $W_{4}/B_{3}$ in such
a way that the $\mathbb{Z}_{2}$-action on $B_{3}$ is exactly that
induced by the complex conjugation operator on the complex algebraic
group $SL\left(5;\mathbb{C}\right)$. Therefore it will fix the compact
real form $SU\left(5\right)$ so that the requisite Wilson line can
be wrapped simultaneously on the $\mathbb{Z}_{2}$-actions on the
Heterotic and $F$-theory sides.

\subsection{Building $B_{3}$ from the action of the Weyl group of $SU\left(5\right)$
on its complexified Cartan subalgebra}

In fact we will build $B_{3}$ as a quotient of the resolution of
the projectivization of the graph of 
\[
\left(\exp\left(2\pi i\text{·}\alpha\right)\mapsto\exp\left(2\pi i\text{·}\left(-\alpha\right)\right)\right)
\]
on the maximal torus 
\[
\exp\left(\mathfrak{h}_{SL\left(5;\mathbb{C}\right)}\right)
\]
by an action of the longest element of the Weyl group $W\left(SL\left(5;\mathbb{C}\right)\right)$.
This will allow the $\mathbb{Z}_{2}$-action on $B_{3}$ to automatically
commute with the action of complex conjugation.

In later sections we will show that $B_{3}$ as constructed will have
the correct numerical characteristics so that the $F$-theory model
(\ref{eq:F-th}) will have the desired properties ($3$-generation,
correct chiral invariants, no vector-like exotics, etc.). The transfer
of information between the $F$-theory and its Heterotic dual is the
subject of a companion paper \cite{Clemens-1}. The application to
the production of the final phenomenologically consistent $F$-theory/Heterotic
duality is the subject of the final paper \cite{Clemens-3} in this
sequence.
\begin{rem}
Throughout this paper, we will let 
\[
\mathbb{P}_{\left[i_{1},\ldots,,i_{d}\right]}^{d-1}
\]
denote the weighted complex projective $\left(d-1\right)$-space with
weights $\left[i_{1},\ldots,,i_{d}\right]$ and will let 
\[
\mathbb{P}_{\left[u_{1},\ldots,u_{d}\right]}
\]
denote the (unweighted) complex projective space with homogeneous
coordinates $\left[u_{1},\ldots,,u_{d}\right]$. 
\end{rem}

\section{The spectral divisor}
\begin{flushleft}
The role of the Tate form (\ref{eq:Tate}) is to break $E_{8}$-symmetry
to that of the first summand of its maximal sub-group 
\[
\frac{SU\left(5\right)_{gauge}\times SU\left(5\right)_{Higgs}}{\mathbb{Z}_{5}}.
\]
The crepant resolution $\tilde{W}_{4}/B_{3}$ of $W_{4}/B_{3}$ will
have $I_{5}$-type fibers over generic points of 
\[
S_{\mathrm{GUT}}:=\left\{ z=0\right\} \subseteq B_{3}.
\]
The $I_{5}$-fibration over $S_{\mathrm{GUT}}$ carries the $SU\left(5\right)_{gauge}$-symmetry.
$SU\left(5\right)_{Higgs}$-symmetry is broken on a five-sheeted branched
covering of $B_{3}$ given by the lift of 
\begin{equation}
\mathcal{C}_{Higgs}:=W_{4}\text{·}\left(\left\{ wy^{2}=x^{3}\right\} -\left\{ w=0\right\} \right)\label{eq:specdiv}
\end{equation}
to a divisor $\tilde{\mathcal{C}}_{Higgs}\subseteq\tilde{W}_{4}$.
Its symmetry is broken by assigning non-trivial eigenvalues to the
fundamental representation $SU\left(5\right)_{Higgs}$ using the spectral
construction with respect to the push-forward to $B_{3}$ of a line
bundle $\mathcal{L}_{Higgs}$ on $\tilde{\mathcal{C}}_{Higgs}$. We
see this as follows. 
\par\end{flushleft}

In parallel to the construction for $SU\left(5\right)_{gauge}$ in
\cite{Clemens-1}, we imbed 
\[
\begin{array}{ccc}
 &  & \mathfrak{h}_{SU\left(5\right)_{Higgs}}^{\mathbb{C}}\\
 &  & \downarrow^{\left(c_{2},c_{3},c_{4},c_{5}\right)}\\
\left(S_{\mathrm{GUT}}-\left\{ a_{0}=0\right\} \right) & \rightarrow & \frac{\mathfrak{h}_{SU\left(5\right)_{Higgs}}^{\mathbb{C}}}{W\left(SU\left(5\right)\right)}
\end{array}
\]
in such a way that the image of $W_{4}-\left\{ a_{0}=0,\,w=1\right\} $
in $\mathbb{C}^{3}\times\frac{\mathfrak{h}_{SU\left(5\right)_{Higgs}}^{\mathbb{C}}}{W\left(SU\left(5\right)\right)}$
is a family of rational double-point surface singularities. The above
diagram allows the Casimir operators $c_{j}$ to operate on the fundamental
representation of $SU\left(5\right)_{Higgs}$ with eigenvalues that
are tracked via a spectral construction \cite{Donagi,Donagi:2008kj}. 
\begin{flushleft}
The Tate form \eqref{eq:Tate}) then records the above geometrically
in $W_{4}/B_{3}$ by considering it as a hypersurface in 
\[
P:=\mathbb{P}\left(\mathcal{O}_{B_{3}}\oplus\mathcal{O}_{B_{3}}\left(2N\right)\oplus\mathcal{O}_{B_{3}}\left(3N\right)\right)
\]
with fiber coordinate $\left[w,x,y\right]$. 
\par\end{flushleft}

\subsubsection{The spectral divisor}
\begin{flushleft}
We define the map
\[
\begin{array}{c}
\mathbb{P}\left(\mathcal{O}_{B_{3}}\oplus\mathcal{O}_{B_{3}}\left(N\right)\right)\rightarrow\mathbb{P}\left(\mathcal{O}_{B_{3}}\oplus\mathcal{O}_{B_{3}}\left(2N\right)\oplus\mathcal{O}_{B_{3}}\left(3N\right)\right)=P\\
\left[w,t\right]\mapsto\left[w,x=t^{2}w,y=t^{3}w\right],
\end{array}
\]
Dividing by $w^{3}$ the inverse image of $W_{4}$ in $\mathbb{P}\left(\mathcal{O}_{B_{3}}\oplus\mathcal{O}_{B_{3}}\left(N\right)\right)$
has equation
\begin{equation}
0=a_{5}t^{5}+a_{4}zt^{4}+a_{3}z^{2}t^{3}+a_{2}z^{3}t^{2}+a_{0}z^{5}\label{eq:ratmp}
\end{equation}
We next blow up up the locus $\left\{ t=z=0\right\} $ in $\mathbb{P}\left(\mathcal{O}_{B_{3}}\oplus\mathcal{O}_{B_{3}}\left(N\right)\right)$
via
\begin{equation}
\left\{ \left|\begin{array}{cc}
t & z\\
T & Z
\end{array}\right|=0\right\} \subseteq\mathbb{P}\left(\mathcal{O}_{B_{3}}\oplus\mathcal{O}_{B_{3}}\left(N\right)\right)\times\mathbb{P}_{\left[T,Z\right]}\label{tee}
\end{equation}
in which the proper transform of \eqref{eq:ratmp} becomes
\begin{equation}
0=a_{5}T^{5}+a_{4}ZT^{4}+a_{3}Z^{2}T^{3}+a_{2}Z^{3}T^{2}+a_{0}Z^{5}\label{eq:speceq2-1}
\end{equation}
defining the \textit{spectral divisor} that we denote as $\mathcal{D}.$
In particular, the spectral divisor expands the singular locus $\left\{ x=y=z=0\right\} $
of $W_{4}$. The condition (\ref{eq:seccond}) implies that homogeneous
form in (\ref{eq:speceq2-1}) is divisible by $Z-T$, that is, the
spectral divisor admits a $\left(4+1\right)$ factorization. 
\begin{equation}
\mathcal{D}=\mathcal{D}^{\left(4\right)}+\mathcal{D}^{\left(1\right)}\subseteq B_{3}\times\mathbb{P}_{\left[T,Z\right]}\label{eq:specdiv-1}
\end{equation}
given by the equation 
\begin{equation}
\begin{array}{c}
0=a_{5}T^{5}+a_{4}ZT^{4}+a_{3}Z^{2}T^{3}+a_{2}Z^{3}T^{2}+a_{0}Z^{5}=\\
\left(a_{5}T^{4}+a_{54}T^{3}Z-a_{20}T^{2}Z^{2}-a_{0}TZ^{3}-a_{0}Z^{4}\right)\left(T-Z\right)
\end{array}\label{eq:speceq-2}
\end{equation}
where $a_{jk}:=a_{j}+a_{k}$. The involution $\tilde{\beta}_{4}/\beta_{3}$
of $W_{4}/B_{3}$ leaves \eqref{eq:speceq-2} and each of its two
factors invariant.\footnote{This $4+1$ split is often written in terms of the variable $s=Z/T$,
e.g. the factor $\left(Z-T\right)$ becomes the factor $\left(s-1\right)$
used to remove $\mathbf{10}{}_{\left\{ -4\right\} }$ states as in
formula (70) in \cite{Blumenhagen-1}. In our case, the $4+1$ split
is global.} 
\par\end{flushleft}

\section{The role of the Cartan sub-algebra $\mathfrak{h}_{SL\left(5;\mathbb{C}\right)_{gauge}}$}

\subsection{Tracking roots}

Again referring to \cite{Clemens-1} we track roots during symmetry-breaking
from $E_{8}$ to the maximal subgroup

\[
\frac{SU\left(5\right)_{gauge}\times SU\left(5\right)_{Higgs}}{\mathbb{Z}_{5}}
\]
on the $F$-theory side by returning to the Tate form 
\[
wy^{2}=x^{3}+a_{5}xyw+a_{4}zx^{2}w+a_{3}z^{2}yw^{2}+a_{2}z^{3}xw^{2}+a_{0}z^{5}w^{3}
\]
where we divide both sides by $a_{0}^{6}$, rescale by the rule 
\[
\begin{array}{c}
\frac{x}{a_{0}^{2}}\rightarrow x\\
\frac{y}{a_{0}^{3}}\rightarrow y\\
\frac{z}{a_{0}}\rightarrow z
\end{array}
\]
and define the functions 
\[
c_{j}=\frac{a_{j}}{a_{0}}
\]
on $B'_{3}:=B_{3}-\left\{ a_{0}=0\right\} $. By this rescaling we
obtain the equation 
\begin{equation}
y^{2}=x^{3}+c_{5}xy+c_{4}zx^{2}+c_{3}z^{2}y+c_{2}z^{3}x+z^{5}\label{eq:singdef}
\end{equation}
in the variables $\left(x,y,z\right)$ parametrized by the `free'
variables $\left(c_{2},c_{3},c_{4},c_{5}\right)$ that we interpret
as a family $\mathcal{V}_{gauge}$ of rational double-point surface
singularities. (See §4.1 of \cite{Clemens-1}.)

Next by interpreting the $c_{j}$ as the $SU\left(5\right)$ Casimir
generators, we pull the family $\mathcal{V}_{gauge}$ back to $\mathfrak{h}_{SL\left(5;\mathbb{C}\right)_{gauge}}$by
the map 
\[
\left(c_{2},c_{3},c_{4},c_{5}\right):\mathfrak{h}_{SL\left(5;\mathbb{C}\right)_{gauge}}\rightarrow\frac{\mathfrak{h}_{SL\left(5;\mathbb{C}\right)_{gauge}}}{W\left(SL\left(5;\mathbb{C}\right)\right)}
\]
where we were able to interpret it as a family of weighted homogeneous
polynomials of weight $30$ that is therefore also obtained via pull-back
from a map to the semi-universal deformation of 
\begin{equation}
y^{2}=x^{3}+z^{5}.\label{eq:E8}
\end{equation}

Next defining 
\begin{equation}
z:=\sum_{j=2}^{5}\kappa_{j}\alpha_{j}\label{eq:zee-1}
\end{equation}
for general complex constants $\kappa_{j}$ as in §4.1 of \cite{Clemens-1},
we obtained morphisms 
\[
B'_{3}\rightarrow\frac{\mathfrak{h}_{SL\left(5;\mathbb{C}\right)}}{W\left(SL\left(5;\mathbb{C}\right)\right)}
\]
and 
\[
W_{4}':=W_{4}\times_{B_{3}}B'_{3}\rightarrow\mathcal{V}_{gauge}
\]
for $B'_{3}=B_{3}-\left\{ a_{0}=0\right\} $. Further we showed that
the complex conjugation operator $\iota$ induces equivariant involutions
\begin{equation}
\left(\left(a_{0},a_{2},a_{3},a_{4},a_{5}\right),x,y,z\right)\mapsto\left(\left(-a_{0},-a_{2},-a_{3},-a_{4},-a_{5}\right),x,-y,-z\right)\label{eq:iota}
\end{equation}
on $W_{4}/B_{3}$ and 
\[
\left(\left(c_{2},c_{3},c_{4},c_{5}\right),x,y,z\right)\mapsto\left(\left(c_{2},-c_{3},c{}_{4},-c_{5}\right),x,-y,z\right)
\]
on (\ref{eq:singdef}). This allowed us in \cite{Clemens-1} to interpret
the equivariant crepant resolution of (\ref{eq:singdef}) over $\mathfrak{h}_{SL\left(5;\mathbb{C}\right)}$
as induced by the Brieskorn-Grothendieck equivariant crepant resolution
\cite{Brieskorn,Slodowy} of the semi-universal deformation of (\ref{eq:E8})
and thereby track roots and the action 
\[
\begin{array}{ccc}
i\text{·}\mathfrak{h}_{SU\left(5\right)_{gauge}}\times i\text{·}\mathfrak{h}_{SU\left(5\right)_{Higgs}} & \overset{-I_{4}\times-I_{4}}{\longrightarrow} & i\text{·}\mathfrak{h}_{SU\left(5\right)_{gauge}}\times i\text{·}\mathfrak{h}_{SU\left(5\right)_{Higgs}}\\
\downarrow &  & \downarrow\\
i\text{·}\mathfrak{h}_{E_{8}} & \overset{-I_{8}}{\longrightarrow} & i\text{·}\mathfrak{h}_{E_{8}}
\end{array}
\]
of complex conjugation.

\subsection{Notation distinguishing Weyl chambers}

We will have a single Tate form defining our $F$-theory model $W_{4}$
but initially we will have two desingularizations that we will denote
as $\dot{W}_{4}/\dot{B}_{3}$ or the `blue' desingularization and
as $\ddot{W}_{4}/\ddot{B}_{3}$ or the `red' desingularization, depending
on whether we consider a given Weyl chamber of $SU\left(5\right)$
or its negative as the `positive' Weyl chamber. $\dot{B}_{3}$ will
be related to $\ddot{B}_{3}$ by a Cremona transformation representing
the passage of each root to its negative. Indeed it is the resolution
of the graph of that transformation that will determine our ultimate
$B_{3}$ and the quotient under its induced involution that will be
our ultimate $B_{3}^{\vee}$.

\section{$S_{4}\subseteq W\left(SU\left(5\right)\right)$ }

Our strategy is now to identify a group $G\leq W\left(SU\left(5\right)\right)$
such that 
\[
\frac{\mathfrak{h}_{SU\left(5\right)}^{\mathbb{C}}}{G}\supseteq B'_{3}\subseteq B_{3}
\]
gives rise to a phenomenologically correct $F$-theory model.

\subsection{Building a toric $B_{3}^{\wedge}$ from $SU\left(5\right)$ roots}

A set of simple positive roots ordered by the $A_{4}$-Dynkin diagram
is given by 
\begin{equation}
\left\{ \alpha_{i}=e_{i}-e_{i-1},\right\} _{i=1,\ldots,4}.\label{eq:pos}
\end{equation}
One immediately checks that the permutation of the axes $e_{j}$ given
by the product of transpositions $\left(e_{0}e_{4}\right)\left(e_{1}e_{3}\right)$
acts as 
\[
\left(\begin{array}{cccc}
\alpha_{1} & \alpha_{2} & \alpha_{3} & \alpha_{4}\end{array}\right)\left(\begin{array}{cccc}
0 & 0 & 0 & -1\\
0 & 0 & -1 & 0\\
0 & -1 & 0 & 0\\
-1 & 0 & 0 & 0
\end{array}\right)=\left(\begin{array}{cccc}
-\alpha_{4} & -\alpha_{3} & -\alpha_{2} & -\alpha_{1}\end{array}\right)
\]
and so is the composition of $-I_{4}$ with the unique symmetry of
the $A_{4}$-Dynkin diagram. Therefore it is the unique longest element
of $W\left(SU\left(5\right)\right)$. This symmetry fixes exactly
one of the five axes, namely the axis $e_{2}$, and therefore lies
the permutation subgroup 
\[
S_{4}=Perm\left\{ e_{0},e_{1},e_{3},e_{4}\right\} \hookrightarrow S_{5}=Perm\left\{ e_{0},e_{1},e_{2},e_{3},e_{4}\right\} =W\left(SU\left(5\right)\right)
\]
via the identification of the axis $e_{j}$ of the fundamental representation
of $SU\left(5\right)$ with the root $e_{j}-e_{2}$. We use this fact
to construct a root basis for $\mathfrak{h}_{SU\left(5\right)}^{\mathbb{C}}$
that is convenient for a toric construction of our `new' $B_{3}$.

We next project the root space along the $e_{2}$-axis to obtain the
vertices of a $3$-dimensional cube. Thus we place $e_{2}$ at $\left(0,0,0\right)$,
the center of the cube below (that we will denote as CUBE) whose eight
vertices are $\left(\text{\textpm}1,\text{\textpm}1,\text{\textpm}1\right)$.
The elements $\left\{ e_{0},e_{1},e_{3},e_{4}\right\} $ can be identified
with the vertices of the blue tetrahedron inscribed in CUBE\smallskip{}
 
\noindent \begin{center}
\includegraphics[width=3cm]{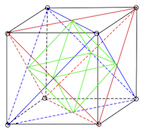} \smallskip{}
 
\par\end{center}

\noindent as follows: 
\[
\begin{array}{cc}
e_{0}: & \left(1,-1,1\right)\\
e_{1:} & \left(1,1,-1\right)\\
e_{3}: & \left(-1,-1,-1\right)\\
e_{4}: & \left(-1,1,1\right)
\end{array}
\]
That is, $\left\{ e_{0},e_{4}\right\} $ are the two `top' blue vertices
and $\left\{ e_{1},e_{3}\right\} $ are the two `bottom' blue vertices
and their negatives are the four vertices of the red tetrahedron.
Circled vertices are those of the polyhedral fan.

\subsection{The Weyl group of $SU\left(5\right)$ and the Cremona involution
as symmetries of the cube}

\noindent The group of orientation-preserving symmetries of CUBE (or
equivalently the orientation-preserving symmetries of the inscribed
green octahedron) maps isomorphically to the permutation group of
axes $\text{\textpm}e_{j}$, that is, 
\[
S_{4}=Perm\left\{ \left\{ \text{\textpm}e_{0}\right\} ,\left\{ \text{\textpm}e_{1}\right\} ,\left\{ \text{\textpm}e_{3}\right\} ,\left\{ \text{\textpm}e_{4}\right\} \right\} \subseteq S_{5}=W\left(SU\left(5\right)\right).
\]

For example, $\left(\left\{ \text{\textpm}e_{0}\right\} \left\{ \text{\textpm}e_{4}\right\} \right)\left(\left\{ \text{\textpm}e_{1}\right\} \left\{ \text{\textpm}e_{3}\right\} \right)$
is the rotation of CUBE around the vertical axis through an angle
of $\pi$. It gives the above longest element of $W\left(SU\left(5\right)\right)$.
Rotation of the cube around the diagonal axis through $\left(1,1,1\right)$
through an angle of $2\pi/3$ is the cyclic permutation $\left(\left\{ \text{\textpm}e_{0}\right\} \left\{ \text{\textpm}e_{1}\right\} \left\{ \text{\textpm}e_{4}\right\} \right)$.
A rotation with axis spanned by the midpoints of a pair of opposite
edges only flips the pair of axes given by the endpoints of the edges.
We will be especially interested in the commutator subgroup of the
involution $\left(\left\{ \text{\textpm}e_{0}\right\} \left\{ \text{\textpm}e_{4}\right\} \right)\left(\left\{ \text{\textpm}e_{1}\right\} \left\{ \text{\textpm}e_{3}\right\} \right)$
in $S_{4}$. 

Finally $A_{4}\subseteq S_{4}$ is the subgroup of orientation-preserving
symmetries of the cube that preserve the blue tetrahedron and therefore
also preserve the red tetrahedron, so that the quotient $S_{4}/A_{4}$
interchanges the two. The full symmetry group of CUBE is then generated
by adjoining the central, orientation-reversing element given by reflection
through the origin that we denote as $C$. It is the involution induced
by $C$ given by $\left(e_{i}-e_{2}\leftrightarrow e_{2}-e_{i}\right)_{i=0,1,3,4}$
on $W_{4}/B_{3}$ that will the yield the quotient $F$-theory model
$W_{4}^{\vee}/B_{3}^{\vee}$ with a $\mathbb{Z}_{2}$-action. It is
the $\mathbb{Z}_{4}$-group generated by the cyclic permutation $\left(\left\{ \text{\textpm}e_{0}\right\} \left\{ \text{\textpm}e_{1}\right\} \left\{ \text{\textpm}e_{4}\right\} \left\{ \text{\textpm}e_{3}\right\} \right)$
that will lead to an asymptotic $\mathbb{Z}_{4}$ $\mathbf{R}$ symmetry
on the semi-stable degeneration $W_{4,0}/B_{3,0}$ of $W_{4}/B_{3}$.
Notice that the Cremona involution $C\notin W\left(SU\left(5\right)\right)$,
however $C$ commutes with all the elements of $S_{4}\leq W\left(SU\left(5\right)\right)$.

\section{Toric geometry of $B_{3}$}

\subsection{Coordinatizing roots}

We choose $U\left(4\right)\subseteq SU\left(5\right)$ via the inclusion
\begin{equation}
\begin{array}{c}
U\left(4\right)\hookrightarrow SU\left(5\right)\\
A\mapsto\hat{A}:=\left(\begin{array}{cc}
\overline{\det A} & 0\\
0 & A
\end{array}\right)
\end{array}\label{eq:subU4}
\end{equation}
with maximal torus of $SU\left(5\right)$ identified with diagonal
unitary matrices $A$. Logarithms of eigenvalue functions for the
restriction of the adjoint representation 
\[
Ad_{SU\left(5\right)}:SU\left(5\right)\rightarrow GL\left(\mathfrak{su}\left(5\right)\otimes\mathbb{C}\right)
\]
to the maximal torus give the (root) linear operators on the complexified
Cartan subalgebra $\mathfrak{h}_{SU\left(5\right)}^{\mathbb{C}}=\mathfrak{h}_{SL\left(5;\mathbb{C}\right)}$.
We choose set of roots $\left\{ e_{j}-e_{2}\right\} _{j=0,1,3,4}$
as basis for $\left(\mathfrak{h}_{SU\left(5\right)}^{\mathbb{C}}\right)^{\ast}$
as in the previous section. Exponentiating we let 
\[
x,y,z,w
\]
denote the corresponding eigenvalue functions.
\begin{rem}
As is standard in the literature, we have used the letter $z$ to
denote the form whose vanishing defines the anti-canonical divisor
$S_{\mathrm{GUT}}\subseteq B_{3}$. In several places below, we will
abuse notation by also using each of the letters $x,y,z,w$ to denote
one of the homogeneous coordinates $\left[x,y,z,w\right]$ of the
$\mathbb{P}^{3}=\mathbb{P}\left(\mathfrak{h}_{A_{4}}^{\mathbb{C}}\right)$
where $\mathfrak{h}_{A_{4}}$ denotes the Cartan subalgebra of $SU\left(5\right)$.
We trust that the intended meaning of $x,y,z,w$ in each instance
will be clear from the context. 
\end{rem}

Then we can make the identification 
\begin{equation}
\begin{array}{c}
\log x=e_{0}-e_{2}\\
\log y=e_{1}-e_{2}\\
\log w=e_{3}-e_{2}\\
\log z=e_{4}-e_{2}.
\end{array}\label{eq:choice}
\end{equation}
giving a basis for the $A_{4}$ root lattice. The distinguished Weyl
chamber (\ref{eq:pos}) is given by the system of positive simple
roots 
\begin{equation}
\begin{array}{c}
\alpha_{1}=e_{1}-e_{0}=\log\left(y/x\right)\\
\alpha_{2}=e_{2}-e_{1}=\log\left(1/y\right)\\
\alpha_{3}=e_{3}-e_{2}=\log w\\
\alpha_{4}=e_{4}-e_{3}=\log\left(z/w\right).
\end{array}\label{eq:choice'}
\end{equation}
(Notice that the set (\ref{eq:choice}) of roots is not a set of simple
roots for a single Weyl chamber, however it does span the root lattice.)
We obtain $24$ of the $120$ Weyl chambers by the $24$ permutations
of $\left\{ x,y,z,w\right\} $ in (\ref{eq:choice'}). The longest
element of the Weyl group is then given by 
\[
\begin{array}{c}
y/x\leftrightarrow w/z\\
1/y\leftrightarrow1/w\\
w\leftrightarrow y\\
z/w\leftrightarrow x/y.
\end{array}
\]

Passing from roots to their negatives corresponds to the Cremona transformation
\begin{equation}
\begin{array}{c}
x\leftrightarrow\frac{1}{x}\\
y\leftrightarrow\frac{1}{y}\\
z\leftrightarrow\frac{1}{z}\\
w\leftrightarrow\frac{1}{w}
\end{array}\label{eq:exactly-1-1}
\end{equation}
that in turn corresponds to the orientation-reversing symmetry of
CUBE given by reflection through the origin.

\subsection{Tracking symmetry-breaking within the Cartan subalgebra of $E_{8}^{\mathbb{C}}$}

As mentioned above, the Tate form tracks symmetry breaking to 
\[
\frac{SU\left(5\right)_{gauge}\times SU\left(5\right)_{Higgs}}{\mathbb{Z}_{5}}\hookrightarrow E_{8}.
\]
As we have shown in \cite{Clemens-1} symmetry-breaking must be compatible
with the three-dimensional commutative diagram

\begin{equation}
\begin{array}{ccccc}
SL\left(5;\mathbb{C}\right) & \overset{\dot{\kappa}}{\hookrightarrow} & \dot{E}_{8}^{\mathbb{C}}\\
\uparrow &  & \uparrow & \nwarrow\\
SU\left(5\right) & \rightarrow & E_{8} &  & \iota\\
\downarrow &  & \downarrow & \swarrow\\
SL\left(5;\mathbb{C}\right) & \overset{\ddot{\kappa}}{\hookrightarrow} & \ddot{E}_{8}^{\mathbb{C}}
\end{array}\label{eq:spin}
\end{equation}
where the top row of (\ref{eq:spin}) is mapped to the bottom row
by the isomorphism $\iota$ given by complex conjugation where $\iota$
induces the involutions 
\begin{equation}
\begin{array}{c}
-I_{4}:\mathfrak{h}_{SL\left(5;\mathbb{C}\right)}\rightarrow\mathfrak{h}_{SL\left(5;\mathbb{C}\right)}\\
-I_{8}:\mathfrak{h}_{E_{8}^{\mathbb{C}}}\rightarrow\mathfrak{h}_{E_{8}^{\mathbb{C}}}
\end{array}\label{eq:Cartan}
\end{equation}
on complex Cartan subalgebras. Identifying maximal tori also identifies
\begin{equation}
\mathfrak{h}_{SL\left(5;\mathbb{C}\right)_{gauge}}\oplus\mathfrak{h}_{SL\left(5;\mathbb{C}\right)_{Higgs}}=\mathfrak{h}_{E_{8}^{\mathbb{C}}}\label{eq:Cartanrel}
\end{equation}
as well as inducing the inclusion 
\[
W\left(SU\left(5\right)_{gauge}\right)\times W\left(SU\left(5\right)_{Higgs}\right)\hookrightarrow W\left(E_{8}\right).
\]

While $-I_{8}$ is also the longest element the Weyl group $W\left(E_{8}\right)$,
it does not restrict to an element of the Weyl group $W\left(SU\left(5\right)_{gauge}\right)$
or an element of the Weyl group $W\left(SU\left(5\right)_{Higgs}\right)$.
However $-I_{8}$ and the pair of longest elements in $W\left(SU\left(5\right)_{gauge}\right)\times W\left(SU\left(5\right)_{Higgs}\right)$
differ by the involutive symmetries of the two $A_{4}$-Dynkin diagrams.

Our strategy will be to build $B_{3}$ and its symmetries from the
group $W\left(SU\left(5\right)_{Higgs}\right)$ acting on on the right-hand
sum in (\ref{eq:Cartanrel}). Throughout we maintain the relationship
with $W\left(E_{8}\right)$ as per (\ref{eq:Cartan}) and (\ref{eq:Cartanrel})
so that the $\mathbb{Z}_{2}$-action on $B_{3}$, that we have denoted
as $\beta_{3}$ in \cite{Clemens-1} and \cite{Clemens-3} and interchangeably
as $C_{u,v}$ below, acts as the composition of the complex conjugation
$-I_{8}$ on \ref{eq:Cartanrel} and the longest element of $W\left(SU\left(5\right)_{Higgs}\right)$
on $B_{3}$.

\subsection{CUBE as a toric polyhedral fan}

The standard toric presentation of $\mathbb{P}^{3}$ is given by a
real vector space $N_{\mathbb{R}}=N_{\mathbb{Z}}\otimes\mathbb{R}=\mathbb{R}^{3}$
with fan equal to the union of four strongly convex rational polyhedral
cones $\sigma_{xyz},\sigma_{xyw},\sigma_{xzw},\sigma_{yzw}\subseteq N_{\mathbb{R}}$
such that for the duals 
\[
S\left(\sigma\right)=\left\{ \mathbf{m}\in\mathrm{Hom}\left(N_{\mathbb{Z}},\mathbb{Z}\right):\mathbf{m}\text{·}\sigma\geq0\right\} 
\]
the associated group algebras $\mathbb{C}\left[S\left(\sigma\right)\right]$
are identified with the respective affine rings $\mathbb{C}\left[x/w,y/w,z/w\right]$,
$\mathbb{C}\left[x/w,y/z,w/z\right]$, $\mathbb{C}\left[x/y,z/y,w/y\right]$,
and $\mathbb{C}\left[y/x,z/x,w/x\right]$. The edges $e_{x},e_{y},e_{z},e_{w}$
of the fan are identified with the four divisors on $\mathbb{P}^{3}$
given by the vanishing of the respective variables.

We have two such toric presentations of $\mathbb{P}^{3}$ in CUBE,
one given by the blue tetrahedron that we will denote as 
\[
\dot{\mathbb{P}}:=\mathbb{P}_{\left[\dot{x},\dot{y},\dot{z},\dot{w}\right]}
\]
and the other given by the red tetrahedron that we will denote as
\[
\ddot{\mathbb{P}}=\mathbb{P}_{\left[\ddot{x},\ddot{y},\ddot{z},\ddot{w}\right]}.
\]
Both toric representations are given with respect to the same toric
lattice $N^{\wedge}$, the one generated by either the red four or
the blue four vertices of CUBE. $\mathbb{P}_{\left[\dot{x},\dot{y},\dot{z},\dot{w}\right]}$
has toric fan given by the vertices of the blue tetrahedron and $\mathbb{P}_{\left[\ddot{x},\ddot{y},\ddot{z},\ddot{w}\right]}$
has toric fan given by the vertices of the red tetrahedron.

This allows us to use CUBE to torically represent the resolution of
the graph of the Cremona transformation (\ref{eq:exactly-1-1}). Namely
the graph of the Cremona transformation is given by the relations
\begin{equation}
\dot{x}\ddot{x}=\dot{y}\ddot{y}=\dot{z}\ddot{z}=\dot{w}\ddot{w}=1\label{eq:Crem2}
\end{equation}
on the Zariski-open subset 
\[
\left\{ \dot{x}\text{·}\dot{y}\text{·}\dot{z}\text{·}\dot{w}\text{\ensuremath{\neq}}0\right\} \times\left\{ \ddot{x}\text{·}\ddot{y}\text{·}\ddot{z}\text{·}\ddot{w}\text{\ensuremath{\neq}}0\right\} \subseteq\mathbb{P}_{\left[\dot{x},\dot{y},\dot{z},\dot{w}\right]}\times\mathbb{P}_{\left[\ddot{x},\ddot{y},\ddot{z},\ddot{w}\right]}.
\]
We define 
\[
B_{3}^{\wedge}:=\left\{ \left(\left[\dot{x},\dot{y},\dot{z},\dot{w}\right],\left[\ddot{x},\ddot{y},\ddot{z},\ddot{w}\right]\right)\in\mathbb{P}_{\left[\dot{x},\dot{y},\dot{z},\dot{w}\right]}\times\mathbb{P}_{\left[\ddot{x},\ddot{y},\ddot{z},\ddot{w}\right]}:\dot{x}\ddot{x}=\dot{y}\ddot{y}=\dot{z}\ddot{z}=\dot{w}\ddot{w}=1\right\} 
\]
as simply the closure of the graph of the Cremona transformation.

$B_{3}^{\wedge}$ is a toric manifold with respect to the same toric
lattice $N^{\wedge}$. The polyhedral fan has vertices at the eight
vertices of CUBE together with the six additional points $\left(\text{\textpm2,0,0}\right)$,
$\left(0,\text{\textpm2,0}\right)$, and $\left(0,0,\text{\textpm2}\right)$.
These fourteen vertices correspond to the fourteen toroidal divisors
whose sum is the anticanonical divisor of $B_{3}^{\wedge}$ . The
inclusion of cones generate two birational morphisms 
\[
\begin{array}{c}
\dot{\pi}:B_{3}^{\wedge}\rightarrow\dot{\mathbb{P}}\\
\ddot{\pi}:B_{3}^{\wedge}\rightarrow\ddot{\mathbb{P}}.
\end{array}
\]
The Cremona involution is given toroidally by the reflection $C$
of CUBE through the origin.

To further describe the toroidal divisors, we denote the red vertices
of the fan as $e_{\dot{x}\ddot{y}\ddot{z}\ddot{w}},e_{\dot{y}\ddot{x}\ddot{z}\ddot{w}},e_{\dot{z}\ddot{x}\ddot{y}\ddot{w}},e_{\dot{w}\ddot{x}\ddot{y}\ddot{z}}$
and the blue vertices as $e_{\dot{x}\dot{y}\dot{z}\ddot{w}},e_{\dot{x}\dot{y}\dot{w}\ddot{z}},e_{\dot{x}\dot{z}\dot{w}\ddot{y}},e_{\dot{y}\dot{z}\dot{w}\ddot{x}}$.
Set 
\[
e_{\dot{x}\dot{z}\ddot{y}\ddot{w}}
\]
the vertex above CUBE that lies on the ray through the origin that
bisects the segment joining $e_{\dot{x}\ddot{y}\ddot{z}\ddot{w}}$
and $e_{\dot{z}\ddot{x}\ddot{y}\ddot{w}}$. This same ray bisects
the segment joining $e_{\dot{x}\dot{y}\dot{z}\ddot{w}}$ and $e_{\dot{x}\dot{z}\dot{w}\ddot{y}}.$
This choice will force the top vertices of the cube to be 
\[
e_{\dot{x}\ddot{y}\ddot{z}\ddot{w}},e_{\dot{x}\dot{y}\dot{z}\ddot{w}},e_{\dot{z}\ddot{x}\ddot{y}\ddot{w}},e_{\dot{x}\dot{z}\dot{w}\ddot{y}}.
\]
We use the analogous notation for the other five possible $\left(2,2\right)$-partitions.
We obtain the toroidal fan:\smallskip{}
 
\begin{center}
\includegraphics[width=8cm]{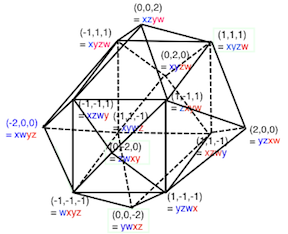} \smallskip{}
\par\end{center}

\begin{flushleft}
The toroidal divisors are given by the vertices of the polyhedral
fan pictured above where the blue-red colorations of the variables
in the monomial 
\[
xyzw
\]
correspond to the decorations $\left\{ \,\dot{},\,\ddot{}\,\right\} $
in $\mathbb{P}_{\left[\dot{x},\dot{y},\dot{z},\dot{w}\right]}\times\mathbb{P}_{\left[\ddot{x},\ddot{y},\ddot{z},\ddot{w}\right]}$. 
\par\end{flushleft}

Passing from roots to their negatives corresponds to the orientation-reversing
symmetry of the above cube given by reflection through the origin.
The reflection also interchanges the blue tetrahedron with the red
one. In fact, the full subgroup $S_{4}\subseteq W\left(SU\left(5\right)\right)$
of symmetries of CUBE acts on the set of decorated monomials, dividing
them into sets with an even number of blue variables and sets with
an odd number of blue variables. The toroidal divisors are just the
restriction to $B_{3}^{\wedge}$ of the divisors 
\begin{equation}
\begin{array}{c}
E_{\dot{x}\ddot{y}\ddot{z}\ddot{w}}:=\left\{ \dot{x}=0\right\} \times\left\{ \ddot{y}=\ddot{z}=\ddot{w}=0\right\} \subseteq\dot{\mathbb{P}}\times\ddot{\mathbb{P}}\\
E_{\dot{x}\dot{z}\ddot{y}\ddot{w}}:=\left\{ \dot{x}=\dot{z}=0\right\} \times\left\{ \ddot{y}=\ddot{w}=0\right\} \subseteq\dot{\mathbb{P}}\times\ddot{\mathbb{P}}\\
E_{\dot{y}\dot{z}\dot{w}\ddot{x}}:=\left\{ \dot{y}=\dot{z}=\dot{w}=0\right\} \times\left\{ \ddot{x}=0\right\} \subseteq\dot{\mathbb{P}}\times\ddot{\mathbb{P}}
\end{array}\label{eq:divisors}
\end{equation}
etc., where of course $E_{\text{·}}$ is the divisor given in toric
notation by $e_{\text{·}}$. For example, the toric dictionary gives
\[
E_{\dot{x}\ddot{y}\ddot{z}\ddot{w}}\leftrightarrow e_{\dot{x}\ddot{y}\ddot{z}\ddot{w}}
\]
and the normal bundle to $E_{\dot{x}\ddot{y}\ddot{z}\ddot{w}}$ in
$\tilde{\mathbb{P}}^{3}$ is 
\[
\mathcal{O}_{E_{\dot{x}\ddot{y}\ddot{z}\ddot{w}}}\left(-1,-1\right).
\]
So by adjunction 
\begin{equation}
K_{\tilde{\mathbb{P}}^{3}}\text{·}E_{\dot{x}\ddot{y}\ddot{z}\ddot{w}}=\mathcal{O}_{E_{\dot{x}\ddot{y}\ddot{z}\ddot{w}}}\left(-1,-1\right).\label{eq:adjunction}
\end{equation}
If we classify the above components (\ref{eq:divisors}) to be of
type $\left(1,3\right)$, $\left(2,2\right)$, and $\left(3,1\right)$
respectively, there are four divisors of type $\left(3,1\right)$,
four divisors of type $\left(1,3\right)$ and six divisors of type
$\left(2,2\right)$. All non-empty intersections occur as intersections
of a component of type $\left(2,2\right)$ with a component of type
$\left(1,3\right)$ or $\left(3,1\right)$ obtained by changing the
decoration on exactly one of its four variables, for example 
\[
E_{\dot{x}\dot{y}\dot{z}\ddot{w}}\cap E_{\dot{x}\dot{y}\ddot{z}\ddot{w}}=\left\{ \dot{x}=\dot{y}=\dot{z}=\ddot{z}=0\right\} .
\]
There are exactly $24=\frac{2\text{·}4!}{2}$ such intersections,
$12$ projecting to a vertex in $\dot{\mathbb{P}}$ and an edge in
$\ddot{\mathbb{P}}$ and $12$ projecting to a vertex $\ddot{\mathbb{P}}$
and an edge in $\dot{\mathbb{P}}$. The divisors of type $\left(3,1\right)$
are the four vertex rays in $N_{\mathbb{R}}^{\wedge}$ in the original
toric description of $\mathbb{P}^{3}$, those of type $\left(1,3\right)$
are the rays through the barycenters of the cones and those of type
$\left(2,2\right)$ are the rays through the barycenters of the faces.

The anti-canonical bundle of $B_{3}^{\wedge}$ is represented by $14$-hedron
given by the (reduced) support of the total transform of the tetrahedron
in $\dot{\mathbb{P}}$ or in $\ddot{\mathbb{P}}$ , that is 
\begin{equation}
\begin{array}{c}
K_{B_{3}^{\wedge}}^{-1}:=E_{\dot{x}\ddot{y}\ddot{z}\ddot{w}}+E_{\dot{y}\ddot{x}\ddot{z}\ddot{w}}+E_{\dot{z}\ddot{x}\ddot{y}\ddot{w}}+E_{\dot{w}\ddot{x}\ddot{y}\ddot{z}}\\
+E_{\dot{x}\dot{z}\ddot{y}\ddot{w}}+E_{\dot{x}\dot{y}\ddot{z}\ddot{w}}+E_{\dot{x}\dot{w}\ddot{y}\ddot{z}}+E_{\dot{y}\dot{z}\ddot{x}\ddot{w}}+E_{\dot{y}\dot{w}\ddot{x}\ddot{z}}+E_{\dot{z}\dot{w}\ddot{x}\ddot{y}}\\
+E_{\dot{y}\dot{z}\dot{w}\ddot{x}}+E_{\dot{x}\dot{z}\dot{w}\ddot{y}}+E_{\dot{x}\dot{y}\dot{w}\ddot{z}}+E_{\dot{x}\dot{y}\dot{z}\ddot{w}}.
\end{array}\label{eq:14gon}
\end{equation}

\subsection{Toric quotients of $B_{3}^{\wedge}$}

Next define the `over-lattice' 
\[
N_{\mathbb{Z}}=\left\{ \left(\frac{a}{2}+\frac{b}{2}+\frac{c}{2}\right):\left(a,b,c\right)\in N^{\wedge},\,a+b+c\equiv_{2}0\right\} \supseteq N^{\wedge}
\]
inducing a toric quotient of $B_{3}^{\wedge}$ . With respect to the
lattice $N_{\mathbb{Z}}$, the polyhedral fan generated by the polyhedral
fan for $B_{3}^{\wedge}$ becomes 
\begin{center}
\smallskip{}
 
\par\end{center}

\begin{center}
\includegraphics[width=3cm]{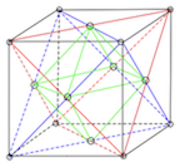} \smallskip{}
 
\par\end{center}

\noindent where the vertices of the fan are circled. 
\begin{flushleft}
Now the six red-blue-green crossing points generate $N_{\mathbb{Z}}$.
The green octahedron with vertices at the six red-blue-green crossing
points \smallskip{}
 
\par\end{flushleft}

\noindent \begin{center}
\includegraphics[width=3cm]{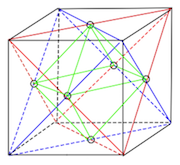} 
\par\end{center}

\noindent \begin{flushleft}
\smallskip{}
 is the toric representation of 
\[
\mathbb{P}_{u,v}:=\mathbb{P}_{\left[u_{0},v_{0}\right]}\times\mathbb{P}_{\left[u_{1},v_{1}\right]}\times\mathbb{P}_{\left[u_{2},v_{2}\right]}
\]
where 
\[
\begin{array}{c}
\left[u_{0},v_{0}\right]=\left[\frac{xz}{yw},\frac{yw}{xz}\right]\\
\left[u_{1},v_{1}\right]=\left[\frac{xy}{zw},\frac{zw}{xy}\right]\\
\left[u_{2},v_{2}\right]=\left[\frac{xw}{yz},\frac{yz}{xw}\right]
\end{array}
\]
and again circled vertices are those of the polyhedral fan. The toric
$\mathbb{P}_{u,v}$ is invariant under the the action of the longest
element $\left(\left(e_{0}e_{4}\right)\left(e_{1}e_{3}\right)\right)$
of $W\left(SU\left(5\right)\right)$, namely the toric involution
given by 
\[
\begin{array}{ccc}
\mathbb{P}_{\left[u_{1},v_{1}\right]}\times\mathbb{P}_{\left[u_{2},v_{2}\right]} & \overset{\left(\left(e_{0}e_{4}\right)\left(e_{1}e_{3}\right)\right)}{\longrightarrow} & \mathbb{P}_{\left[u_{1},v_{1}\right]}\times\mathbb{P}_{\left[u_{2},v_{2}\right]}\\
\left(\left[u_{1},v_{1}\right],\left[u_{2},v_{2}\right]\right) & \mapsto & \left(\left[v_{1},u_{1}\right],\left[v_{2},u_{2}\right]\right).
\end{array}
\]
\par\end{flushleft}

We will distinguish the 'vertical' $\mathbb{P}^{1}=\mathbb{P}_{\left[u_{0},v_{0}\right]}$.
Thus $\left\{ u_{0}=0\right\} $ will correspond to the 'top' of the
cube and $\left\{ v_{0}=0\right\} $ will correspond to the 'bottom.'
Our distinguished 'vertical' $\mathbb{P}_{\left[u_{0},v_{0}\right]}$
in CUBE is the one spanned by 
\[
f_{\dot{x}\dot{z}\ddot{y}\ddot{w}}:=\frac{1}{2}\text{·}e_{\dot{x}\dot{z}\ddot{y}\ddot{w}}\,and\,f_{\dot{y}\dot{w}\ddot{x}\ddot{z}}=-f_{\dot{x}\dot{z}\ddot{y}\ddot{w}}
\]
that is, this $\mathbb{P}^{1}$ is the one corresponding to the partition
$\left(\left\{ x,z\right\} ,\left\{ y,w\right\} \right)$ of $\left\{ x,y,z,w\right\} $.
Analogously we have $\mathbb{P}_{\left[u_{1},v_{1}\right]}$ corresponding
to the fan vertices 
\[
f_{\dot{x}\dot{y}\ddot{z}\ddot{w}}\,and\,f_{\dot{z}\dot{w}\ddot{x}\ddot{y}}=-f_{\dot{x}\dot{y}\ddot{z}\ddot{w}},
\]
and finally $P_{\left[u_{2},v_{2}\right]}$ corresponding to the fan
\[
f_{\dot{x}\dot{w}\ddot{y}\ddot{z}}\,and\,f_{\dot{y}\dot{z}\ddot{x}\ddot{w}}=-f_{\dot{x}\dot{w}\ddot{y}\ddot{z}}.
\]

The following Lemma then follows immediately from the toroidal picture
together with the fact that the Cremona involution on $B_{3}^{\wedge}$
reverses the decorations on the $xyzw$-monomials, whereas the action
of $S_{4}$ permutes the variables $x,y,z,w$. 
\begin{lem}
\label{lem:i)-The-involution}i) The involution 
\[
C_{u,v}:\mathbb{P}_{u,v}\rightarrow\mathbb{P}_{u,v}
\]
induced by the Cremona involution $C$ on $B_{3}^{\wedge}$ is given
by 
\[
\left[u_{j},v_{j}\right]\mapsto\left[v_{j},u_{j}\right]
\]
for $j=0,1,2$.

ii) The $\mathbb{Z}_{4}$-action by the cyclic permutation
\begin{equation}
T_{0}:=\left(\left\{ \text{\textpm}e_{0}\right\} \left\{ \text{\textpm}e_{1}\right\} \left\{ \text{\textpm}e_{4}\right\} \left\{ \text{\textpm}e_{3}\right\} \right)\label{eq:tee}
\end{equation}
is a rotation of $90\text{°}$ of the green octahedron above in this
Subsection around its vertical axis. 

iii) The cyclic permutation $\left(\left\{ \text{\textpm}e_{0}\right\} \left\{ \text{\textpm}e_{1}\right\} \left\{ \text{\textpm}e_{4}\right\} \left\{ \text{\textpm}e_{3}\right\} \right)$
is the $\mathbb{Z}_{4}$-action on $\mathbb{P}_{\left[u_{1},v_{1}\right]}\times\mathbb{P}_{\left[u_{2},v_{2}\right]}$
generated by
\[
\left[u_{1},v_{1}\right],\left[u_{2},v_{2}\right]\overset{T_{0}}{\longrightarrow}\left[u_{2},v_{2}\right],\left[v_{1},u_{1}\right]
\]
that we also denote as $T_{0}$. 

iv) Replacing the homogeneous coordinates $\left[u_{0},v_{0}\right]$
for $\mathbb{P}_{\left[u_{0},v_{0}\right]}$ with the single affine
coordinate $\frac{u_{0}-v_{0}}{u_{0}+v_{0}}$, there exists a $\mathbb{Z}_{4}$-action
that we denote as $T_{u,v}$ on $B_{3}=\mathbb{P}_{\left[u_{0},v_{0}\right]}\times B_{2}$
defined by
\begin{equation}
\left(\frac{u_{0}-v_{0}}{u_{0}+v_{0}},\left[u_{1},v_{1}\right],\left[u_{2},v_{2}\right]\right)\overset{T_{u,v}}{\longrightarrow}\left(i\text{·}\left(\frac{u_{0}-v_{0}}{u_{0}+v_{0}}\right),\,T_{0}\left(\left[u_{1},v_{1}\right],\left[u_{2},v_{2}\right]\right)\right).\label{eq:cee}
\end{equation}

iii) Furthermore
\end{lem}

\begin{equation}
\begin{array}{c}
T_{u,v}^{2}\left(\frac{u_{0}-v_{0}}{u_{0}+v_{0}},\left[u_{1},v_{1}\right],\left[u_{2},v_{2}\right]\right)=\left(\frac{v_{0}-u_{0}}{u_{0}+v_{0}},\left[v_{1},u_{1}\right],\left[v_{2},u_{2}\right]\right).\\
=C_{u,v}\left(\frac{u_{0}-v_{0}}{u_{0}+v_{0}},\left[u_{1},v_{1}\right],\left[u_{2},v_{2}\right]\right).
\end{array}.\label{eq:comm2}
\end{equation}

\section{\label{sec:,--and}$B_{2}$, $B_{3}$ and their symmetries}

The toroidal blow-up of $\mathbb{P}_{u,v}$ at the eight points of
\[
\left\{ u_{0}v_{0}=u_{1}v_{1}=u_{2}v_{2}=0\right\} 
\]
as shown in the previous Section cannot be chosen for $B_{3}$ since
its anti-canonical line bundle is far from ample. In fact it is given
by sections of $\mathcal{\mathcal{O}}_{\mathbb{P}_{u,v}}\left(2,2,2\right)$
that vanish to second order at the eight vertices $\left\{ u_{0}v_{0}=u_{1}v_{1}=u_{2}v_{2}=0\right\} $
and therefore also vanish to first order along all the edges of CUBE.
So its anti-canonical linear system will not be basepoint-free which
will be necessary for our geometric model. Furthermore it has 
\[
\left(K^{-1}\right)^{3}=-16.
\]

On the other hand $K_{\mathbb{P}_{u,v}}^{-1}$ is ample with 
\[
\left(K_{\mathbb{P}_{u,v}}^{-1}\right)^{3}=48.
\]
In our application to $W_{4}/B_{3}$ we will need the linear system
$\left|K_{B_{3}}^{-1}\right|$ basepoint-free and, in particular,
for three-generation we will need 
\begin{equation}
\left(K_{B_{3}}^{-1}\right)^{3}=12.\label{eq:justright}
\end{equation}

To achieve (\ref{eq:justright}), we will modify 
\[
\mathbb{P}_{B_{2}}:=\mathbb{P}_{\left[u_{1},v_{1}\right]}\times\mathbb{P}_{\left[u_{2},v_{2}\right]}.
\]
Since each blown up point on $\mathbb{P}_{B_{2}}$ reduces $\left(K_{B_{3}}^{-1}\right)^{3}$
by six, we will need to blow up $\mathbb{P}_{B_{2}}$ at six points.
Furthermore, in order that the $F$-theory model incorporate an eventual
$\mathbb{Z}_{4}$-action that induces asymptotically a $\mathbb{Z}_{4}$
$\mathbf{R}$-symmetry, the six points will have to comprise the union
of two orbits of the $\mathbb{Z}_{4}$-action. 

The fan of the toric representation of $\mathbb{P}_{B_{2}}$ is the
horizontal square of the green octohedron in the above figures also
shown as the tilted square in the diagram below. We blow up $\mathbb{P}_{B_{2}}$
torically by adjoining the vertices $\left(1,1\right)$ and $\left(-1,-1\right)$
to its fan to obtain the toric fan of $D_{6}$. But this polyhedral
fan can also be viewed as the toric fan $\mathbb{P}^{2}=\mathbb{P}_{\left[a,b,c\right]}$
blown up at the three points with only one non-zero coordinate. It
is obtained by adjoining the three vertices $\left(1,1\right)$,$\left(-1,0\right)$
and $\left(0,-1\right)$ to the isosceles triangle fan of $\mathbb{P}_{\left[a,b,c\right]}$.

\smallskip{}

\noindent \begin{center}
\includegraphics[scale=0.4]{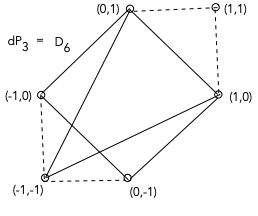}
\par\end{center}

\noindent \begin{flushleft}
\smallskip{}
In fact this isomorphism is given explicitly by the correspondence
\par\end{flushleft}

\noindent \begin{center}
\smallskip{}
\begin{tabular}{|c|c|c|}
\hline 
$\mathbb{P}_{\left[a,b,c\right]}$ & $\mathbb{P}_{B_{2}}$ & Fan vertex\tabularnewline
\hline 
\hline 
blow up $\left\{ a=b=0\right\} $ & blow up $\left\{ u_{1}=u_{2}=0\right\} $ & $\left(1,1\right)$\tabularnewline
\hline 
proper transform $\left\{ a=0\right\} $ & proper transform $\left\{ u_{2}=0\right\} $ & $\left(1,0\right)$\tabularnewline
\hline 
blow up $\left\{ a=c=0\right\} $ & proper transform $\left\{ v_{1}=0\right\} $ & $\left(0,-1\right)$\tabularnewline
\hline 
proper transform $\left\{ c=0\right\} $ & blow up $\left\{ v_{1},v_{2}=0\right\} $ & $\left(-1,-1\right)$\tabularnewline
\hline 
blow up $\left\{ b=c=0\right\} $ & proper transform $\left\{ v_{2}=0\right\} $ & $\left(-1,0\right)$\tabularnewline
\hline 
proper transform $\left\{ b=0\right\} $ & proper transform $\left\{ u_{1}=0\right\} $ & $\left(0,1\right)$\tabularnewline
\hline 
\end{tabular}\smallskip{}
\par\end{center}

We next analyze in detail the construction of the del Pezzo surface
$B_{2}$ for $B_{3}=\mathbb{P}_{\left[u_{0},v_{0}\right]}\times B_{2}$. 

\subsection{The del Pezzo $B_{2}$ and its symmetries}

To achieve a Fano $B_{3}=B_{2}\times\mathbb{P}_{\left[u_{0},v_{0}\right]}$
necessary for the $F$-theory dual, $B_{2}$ must be a del Pezzo surface.
By Castelnuovo's Rationality Theorem, the $\mathbb{Z}_{2}$-action
on $B_{2}$, that we call $\beta_{2}$, cannot be free. On the other
hand, $\beta_{2}$ must have at most finite fixpoint set since otherwise
the $\mathbb{Z}_{2}$-action $\beta_{3}$ on $B_{3}$ cannot act freely
on the ample anti-canonical section $S_{\mathrm{GUT}}\subseteq B_{3}$,
a necessary condition for an $F$-theory model with Wilson-line symmetry
breaking. Now by Table 6 and Figure 10 in \cite{Blumenhagen}, there
is one and only one sequence of del Pezzo surfaces with involution
having finite fixpoint set. These are the four entries in Table 6
that have no entry in either the $\Sigma$ column nor in the $\mathbf{R}$
column. The sequence is represented in Figure 10 by the left vertical
column that begins with $\mathbb{P}_{\left[u_{1},v_{1}\right]}\times\mathbb{P}_{\left[u_{2},v_{2}\right]}$
and proceeds by blowing up three additional pairs of points to obtain
the phenomenologically desirable $dP_{7}=D_{2}$, the standard mathematical
notation for the family of del Pezzo surfaces whose anti-canonical
divisor has self-intersection $2$.\footnote{More precisely, 'phenomenologically desirable' equates to '$3$-generation,
one Higgs doublet, and no vector-like exotics.' } 

We will also need the $\mathbb{Z}_{4}$-symmetry inherited from a
square root of the longest element of the Weyl group of $SU\left(5\right)$
as in Lemma \ref{lem:i)-The-involution} $T_{0}$ acts as 
\begin{equation}
u_{1}\mapsto u_{2}\mapsto v_{1}\mapsto v_{2}\mapsto u_{1}\label{eq:4sym}
\end{equation}
on the blown up $\mathbb{P}_{\left[u_{1},v_{1}\right]}\times\mathbb{P}_{\left[u_{2},v_{2}\right]}$.Whether
or not $B_{2}$ is del Pezzo will, as we see next, depends on on the
choice of the orbit of the $\mathbb{Z}_{4}$-symmetry on $\mathbb{P}_{\left[u_{1},v_{1}\right]}\times\mathbb{P}_{\left[u_{2},v_{2}\right]}$
generated by $T_{0}$, a cyclic element of the Weyl group of $SU\left(5\right)$
whose square is its longest element. 
\begin{thm}
\label{thm:For-generic-choice}For generic choice of the orbit of
the action \eqref{eq:4sym}, the resulting $B_{2}$ is a del Pezzo
surface.
\end{thm}

\begin{proof}
The proof will first transfer the $\mathbb{Z}_{4}$-action from $\mathbb{P}_{\left[u_{1},v_{1}\right]}\times\mathbb{P}_{\left[u_{2},v_{2}\right]}$
to $\mathbb{P}_{\left[a,b,c\right]}$ and then check general position
for points of an orbit, that is, no three points of the seven blown
up point lie on a line and no six points lie on a conic. The birational
passage from $\mathbb{P}_{\left[u_{1},v_{1}\right]}\times\mathbb{P}_{\left[u_{2},v_{2}\right]}$
to $\mathbb{P}_{\left[a,b,c\right]}$ is given as follows. An orbit
of the $\mathbb{Z}_{4}$-action \eqref{eq:4sym} on $\mathbb{P}_{\left[u_{1},v_{1}\right]}\times\mathbb{P}_{\left[u_{2},v_{2}\right]}$
can be written as
\[
\begin{array}{c}
\left(\left[u_{1},v_{1}\right],\left[u_{2},v_{2}\right]\right)=\left(\left[\frac{u_{1}}{v_{1}},1\right],\left[\frac{u_{2}}{v_{2}},1\right]\right)\\
\left(\left[u_{2},v_{2}\right],\left[v_{1},u_{1}\right]\right)=\left(\left[\frac{u_{2}}{v_{2}},1\right],\left[\frac{v_{1}}{u_{1}},1\right]\right)\\
\left(\left[v_{1},u_{1}\right],\left[v_{2},u_{2}\right]\right)=\left(\left[\frac{v_{1}}{u_{1}},1\right],\left[\frac{v_{2}}{u_{2}},1\right]\right)\\
\left(\left[v_{2},u_{2}\right],\left[u_{1},v_{1}\right]\right)=\left(\left[\frac{v_{2}}{u_{2}},1\right],\left[\frac{u_{1}}{v_{1}},1\right]\right)
\end{array}
\]
that translates to 
\begin{equation}
\begin{array}{c}
\left(a,b\right)\\
\left(b,a^{-1}\right)\\
\left(a^{-1},b^{-1}\right)\\
\left(b^{-1},a\right)
\end{array}\label{eq:4aff}
\end{equation}
that together with $\left(0,0\right)$, $\left(\infty,0\right)$ and
$\left(0,\infty\right)$ become the seven blown-up points in $\mathbb{P}_{\left[a,b,c\right]}$
by setting $c=1$. One immediately checks that for general choice
of $a$ and $b$ the slopes of the four points in \eqref{eq:4aff}
are distinct and non-zero. So no line containing two of the points
that are blown up contains any of the $5$ other points that are blown
up. Also conics through $6$ of the $7$ points fall into one of the
two cases: 

1) Parabolas with horizontal or vertical major axis containing $\left(\left(0.0\right)\right)$
and so of the form

\begin{equation}
y=cx\left(x-2d\right)\label{eq:4orbit}
\end{equation}
containing all four points of \eqref{eq:4aff}. Therefore 
\[
\begin{array}{c}
\left(a,b\right):\,\,b=ca\left(a-2d\right)\,\,c=\frac{b}{a\left(a-2d\right)}\\
\left(b,a^{-1}\right):\,\,a^{-1}=cb\left(b-2d\right)\,\,c=\frac{a^{-1}}{b\left(b-2d\right)}\\
\left(a^{-1},b^{-1}\right):\,\,b^{-1}=ca^{-1}\left(a^{-1}-2d\right)\,\,c=\frac{b^{-1}}{a^{-1}\left(a^{-1}-2d\right)}\\
\left(b^{-1},a\right):\,\,a=cb^{-1}\left(b^{-1}-2d\right)\,\,c=\frac{a}{b^{-1}\left(b^{-1}-2d\right)}
\end{array}
\]
so by high school algebra we obtain a polynomial relationship between
$a$ and $b$
\[
\begin{array}{c}
\frac{b}{a\left(a-2d\right)}=\frac{a^{-1}}{b\left(b-2d\right)}:\,\,\frac{b^{3}-a}{b^{2}-1}=2d\\
\frac{b}{a\left(a-2d\right)}=\frac{b^{-1}}{a^{-1}\left(a^{-1}-2d\right)}:\,\,\frac{ba^{-2}-b^{-1}a^{2}}{ba^{-1}-b^{-1}a}=2d\\
a\left(b^{3}-a\right)\left(b^{2}-a^{2}\right)=\left(b^{2}-1\right)\left(b^{2}-a^{4}\right).
\end{array}
\]

2) Hyperbolae with the two coordinate axes as asymptotes and so of
the form
\begin{equation}
\left(x-c\right)\left(y-d\right)=e\label{eq:inj1}
\end{equation}
containing all four points of \eqref{eq:4aff} or 
\begin{equation}
xy-\left(cy+dx\right)=0\label{eq:bdc2}
\end{equation}
containing three four points of \eqref{eq:4aff}. In the case of \eqref{eq:inj1}
\[
\begin{array}{c}
\left(a,b\right):\,\,\left(a-c\right)\left(b-d\right)=e\\
\left(b,a^{-1}\right):\,\,\left(b-c\right)\left(1-ad\right)=ae\\
\left(a^{-1},b^{-1}\right):\,\,\left(1-ac\right)\left(1-bd\right)=abe\\
\left(b^{-1},a\right):\,\,\left(1-bc\right)\left(a-d\right)=be.
\end{array}
\]
So
\[
\begin{array}{c}
a\left(a-c\right)\left(b-d\right)=\left(b-c\right)\left(1-ad\right)\\
ab\left(a-c\right)\left(b-d\right)=\left(1-ac\right)\left(1-bd\right)\\
b\left(a-c\right)\left(b-d\right)=\left(1-bc\right)\left(a-d\right)
\end{array}
\]
and so
\[
\begin{array}{c}
d=\frac{\left(b-c\right)-a\left(a-c\right)b}{\left(b-c\right)a-a\left(a-c\right)}\\
d=\frac{1-ac-ab\left(a-c\right)b}{\left(1-ac\right)b-ab\left(a-c\right)}\\
d=\frac{a\left(1-bc\right)-b^{2}\left(a-c\right)}{\left(1-bc\right)-b\left(a-c\right)}.
\end{array}
\]
As above this gives us two linear equations in $c$ with coefficients
that are polynomials in $\left(a,b\right)$. Again as above solving
each for $c$and setting the two expressions for $c$ equal to each
other we get a non-trivial polynomial equation in $\left(a,b\right)$
that must be satisfied so that this conic passes through the six given
points. The case \eqref{eq:bdc2} is similar. We have three linear
equations in $\left(c,d\right)$. Use pairs of these to eliminate
$d$ and then set the two expression for $c$ equal to each other
to get a non-trivial polynomial relation in $\left(a,b\right)$ that
must be satisfied.

We have therefore shown that if $\left(a,b\right)$ does not satisfy
any of the above finite number of polynomial equations, then $D_{6}$
blown up at the orbit of $\left(a,b\right)$ under the $\mathbb{Z}_{4}$-action
is a del Pezzo surface.
\end{proof}
Thus for generic choice of $\left(a,b\right)$ the resulting $dP_{7}$
is in fact a del Pezzo surface $D_{2}$. We define
\[
B_{2}=D_{2}
\]
with 
\[
B_{3}=\mathbb{P}_{\left[u_{0},v_{0}\right]}\times B_{2}.
\]
Then we will consider the specialization of \eqref{eq:4sym} under
the specialization
\[
\left(B_{3}=\mathbb{P}_{\left[u_{0},v_{0}\right]}\times B_{2}\right)\Rightarrow\left(B_{3,0}=\left(\mathbb{P}_{\left[1,a\right]}\cup\mathbb{P}_{\left[1,b\right]}\right)\times B_{2}\right)
\]
to the semi-stable degeneration such that the limit that encodes the
structure of the Heterotic model. $T_{u,v}$ will act equivariantly
on our degenerating family of Calabi-Yau fourfolds that we will denote
$W_{4,\delta}/B_{3,\delta}$ with $\delta$ denoting the parameter
for the degenerating family of fourfolds. Actually for all $\delta\neq0$
\[
B_{3}=B_{2}\times\mathbb{P}_{\left[u_{0},v_{0}\right]}\cong B_{3,\delta}
\]
but at $\delta=0$ , while $B_{2}$ remains stationary, we will have
the $\mathbb{P}^{1}$-splitting
\[
\mathbb{P}_{\left[u_{0},v_{0}\right]}\Rightarrow\mathbb{P}_{\left[1,a\right]}\cup\mathbb{P}_{\left[1,b\right]}
\]
that will force $W_{4,0}$ to split into two components $dP_{a}/\left(B_{2}\times\mathbb{P}_{\left[1,a\right]}\right)$
and $dP_{b}/\left(B_{2}\times\mathbb{P}_{\left[1,b\right]}\right)$.
We will thereby equip the semi-stable limit 
\begin{equation}
\begin{array}{c}
W_{4,0}/B_{2}=\left(dP_{a}\cup dP_{b}\right)/B_{2}\\
V_{3}/B_{2}=\left(dP_{a}\cap dP_{b}\right)/B_{2}
\end{array}\label{eq:semi}
\end{equation}
with an 'asymptotic $\mathbb{Z}_{4}$ \textbf{R}-symmetry.' The action
of the element of order $2$ in $\mathbb{Z}_{4}$ is simply the $\mathbb{Z}_{2}$-action
$\beta_{3}$ employed in \eqref{eq:F-th} and papers \cite{Clemens-1,Clemens-3}.
This passage to the semistable limit will be defined and described
in detail in Section \ref{subsec:Modification-for-} of this paper.

\subsection{Eigenvectors and eigenvalues for the $\mathbb{Z}_{4}$ symmetry on
$B_{3}$}

By Theorem \ref{thm:For-generic-choice}, blowing up $D_{6}$ at the
four additional points \eqref{eq:4aff} of a generic orbit of $T_{0}$
yields a del Pezzo $D_{2}$ if the four points are in general position.
Since $D_{r}$ for $r<6$ is no longer toric, we have to leave the
family of toric varieties in order to achieve a smooth threefold with
the correct numerical invariants. Now four points on $\mathbb{P}_{\left[a,b,c\right]}$
in general position are the base locus of a pencil of conics on which
$T_{0}$ also acts. This action has two fixpoints $q_{1}\left(a,b,c\right)$
and $q_{2}\left(a,b,c\right)$ and the blow-up of $D_{6}$ at the
four points given us by the smooth divisor
\begin{equation}
\left\{ \left|\begin{array}{cc}
q_{1} & q_{2}\\
k & l
\end{array}\right|=0\right\} \subseteq D_{6}\times\mathbb{P}_{\left[k,l\right]}.\label{eq:4blow}
\end{equation}

$T_{0}$ also acts equivariantly on the mapping
\begin{equation}
B_{2}\rightarrow\mathbb{P}\left(H^{0}\left(K_{B_{2}}^{-1}\right)\right)\cong\mathbb{P}^{2}\label{eq:tB}
\end{equation}
that turns out to be a double cover with branch locus a smooth quartic
plane curve $R\subseteq\mathbb{P}\left(H^{0}\left(K_{B_{2}}^{-1}\right)\right)$.
We define $q_{\left[1,0,0\right]},q_{\left[0,1,0\right]},q_{\left[0,0,1\right]}$
as the conic in the pencil \eqref{eq:4blow} containing the point
indicated by its respective subscript. Then
\[
\begin{array}{c}
a\cdot q_{\left[1,0,0\right]}\\
b\cdot q_{\left[0,1,0\right]}\\
c\cdot q_{\left[0,0,1\right]}
\end{array}
\]
are a basis for the anti-canonical linear system of $B_{2}$. It turns
out that 
\[
\begin{array}{c}
t_{1}=\log a\\
t_{2}=\log b\\
t_{3}=t_{1}=\log c
\end{array}
\]
are natual coordinates for $H^{0}\left(K_{B_{2}}^{-1}\right)$ so
that we desire the choice of $B_{2}$ to be such that , referring
to \eqref{eq:4aff}, we look for 
\[
B_{2}=\left\{ t_{0}^{2}=f_{4}\left(t_{1},t_{2},t_{3}\right)\right\} 
\]
such that the generator $T_{0}$of the $\mathbb{Z}_{4}$-action
\[
\begin{array}{c}
\left(\log a,\log b\right)\\
\left(\log b,-\log a\right)\\
\left(-\log a,-\log b\right)\\
\left(-\log b,\log a\right),
\end{array}
\]
that is
\[
\begin{array}{c}
\left(t_{0};\left[t_{1},t_{2},t_{3}\right]\right)\mapsto\\
\left(\det\left(\begin{array}{ccc}
0 & -1 & 0\\
1 & 0 & 0\\
0 & 0 & -1
\end{array}\right)\cdot t_{0};\left[\left(\begin{array}{ccc}
t_{1} & t_{2} & t_{3}\end{array}\right)\left(\begin{array}{ccc}
0 & -1 & 0\\
1 & 0 & 0\\
0 & 0 & -1
\end{array}\right)\right]\right)
\end{array}
\]
is an automorphism of $B_{2}$ with characteristic polynomial 
\[
\frac{\lambda^{4}-1}{\lambda-1}=\lambda^{3}+\lambda^{2}+\lambda+1.
\]
In fact we can demand that $\mathrm{Aut}B_{2}$ is the entire symmetry
group of CUBE, namely 
\[
S_{4}\times\mathbb{Z}_{2}=signed\,S_{3},
\]
where 
\[
B_{2}=\left\{ t_{0}^{2}=t_{1}^{4}+t_{2}^{4}+t_{3}^{4}+\alpha\left(t_{1}^{2}t_{2}^{2}+t_{1}^{2}t_{3}^{2}+t_{2}^{2}t_{3}^{2}\right)\right\} .
\]

By elementary algebra bitangent lines to $R$ that we will denote
as $\left\{ m_{1}:=c_{\alpha}t_{1}+t_{2}=0\right\} $ and $\left\{ m_{2}:=t_{1}-c_{\alpha}t_{2}=0\right\} $
respectively. Then
\[
\begin{array}{c}
m_{1}\circ T_{0}=m_{2}\\
m_{2}\circ T_{0}=-m_{1}.
\end{array}
\]
with eigenvectors
\[
\begin{array}{c}
m_{+i}:=m_{1}-i\text{·}m_{2}\\
m_{-i}:=m_{1}+i\text{·}m_{2}
\end{array}
\]
indexed by eigenvalues for the action of $T_{0}$ on $H^{0}\left(K_{B_{2}}^{-1}\right)$.
The third eigenvalue must be $-1$. Let $n_{-1}$ denote the associated
eigenvector. In what follows and associated paper, it will be convenient
to appeal to two sets of coordinates for $\mathbb{P}\left(H^{0}\left(K_{B_{2}}^{-1}\right)\right)$,
one being the coordinates $\mathbb{P}_{\left[n_{-1},m_{+i},m_{-i}\right]}$
given by eigenvectors and the other being $\mathbb{P}_{\left[n_{0},m_{1},m_{2}\right]}$
incorporating the defining forms for the bitangent lines defined above.
To avoid notational confusion 
\[
n_{0}=n_{-1}
\]
defines the same line in $\mathbb{P}\left(H^{0}\left(K_{B_{2}}^{-1}\right)\right)$
but the different index indicates which coordinate system we are referring
to. Then

\begin{equation}
\begin{array}{c}
T_{0}^{\ast}\left(n_{0}\right)=n_{0}\circ T_{0}=-n_{0}\\
T_{0}^{\ast}\left(m_{1}\right)=m_{1}\circ T_{0}=m_{2}\\
T_{0}^{\ast}\left(m_{2}\right)=m_{2}\circ T_{0}=-m_{1}
\end{array}\label{eq:t0}
\end{equation}
and
\begin{equation}
\begin{array}{c}
T_{0}^{\ast}\left(n_{-1}\right)=n_{-1}\circ T_{0}=-n_{0}\\
T_{0}^{\ast}\left(m_{+i}\right)=m_{+i}\circ T_{0}=i\text{·}m_{+i}\\
T_{0}^{\ast}\left(m_{-i}\right)=m_{-i}\circ T_{0}=-i\text{·}m_{1}.
\end{array}\label{eq:t1}
\end{equation}
 With respect to the direct-sum decomposition 
\begin{equation}
H^{0}\left(K_{B_{2}}^{-1}\right)=H^{0}\left(K_{B_{2}}^{-1}\right)^{\left[+1\right]}\oplus H^{0}\left(K_{B_{2}}^{-1}\right)^{\left[-1\right]}\label{eq:pushsplit}
\end{equation}
induced by the involution $T_{B_{2}}^{2}$ the first summand is generated
by a single section $n_{0}$ while the second summand is the two-dimensional
$T_{0}^{2}$-eigenspace with eigenvalue $-1$. The second summand
is spanned by any two of the four vectors
\[
m_{1},m_{2},m_{+i},m_{-i}.
\]

For $B_{3}=B_{2}\times\mathbb{P}_{\left[u_{0},v_{0}\right]}$, the
restriction map
\begin{equation}
H^{0}\left(\mathcal{O}_{\mathbb{P}_{\left[n_{0},m_{1},m_{2}\right]}}\left(1\right)\boxtimes\mathcal{O}_{\mathbb{P}_{\left[u_{0},v_{0}\right]}}\left(2\right)\right)\rightarrow H^{0}\left(K_{B_{3}}^{-1}\right)\label{eq:prod}
\end{equation}
is an isomorphism. Therefore we can write
\[
T_{u,v}\left(\frac{u_{0}-v_{0}}{u_{0}+v_{0}},\left(\left[n_{-1},m_{+i},m_{-i}\right]\right)\right)=\left(i\text{·}\left(\frac{u_{0}-v_{0}}{u_{0}+v_{0}}\right),T_{0}\left(\left[n_{-1},m_{+i},m_{-i}\right]\right)\right)
\]
on $H^{0}\left(\mathcal{O}_{\mathbb{P}_{\left[n_{-1},m_{+i},m_{-i}\right]}}\left(1\right)\boxtimes\mathcal{O}_{\mathbb{P}_{\left[u_{0},v_{0}\right]}}\left(2\right)\right)$
or equivalently 
\[
T_{u,v}\left(\frac{u_{0}-v_{0}}{u_{0}+v_{0}},\left(\left[n_{0},m_{1},m_{2}\right]\right)\right)=\left(i\text{·}\left(\frac{u_{0}-v_{0}}{u_{0}+v_{0}}\right),T_{0}\left(\left[n_{0},m_{1},m_{2}\right]\right)\right)
\]
on $H^{0}\left(\mathcal{O}_{\mathbb{P}_{\left[n_{0},m_{1},m_{2}\right]}}\left(1\right)\boxtimes\mathcal{O}_{\mathbb{P}_{\left[u_{0},v_{0}\right]}}\left(2\right)\right)$.
We sum up as follows.
\begin{lem}
\label{lem:i)-The-sub-linear32} i) $\left|K_{\mathbb{P}_{\left[u_{0},v_{0}\right]}\times B_{2}}^{-1}\right|$
is basepoint-free. Also 
\[
h^{0}\left(K_{\mathbb{P}_{\left[u_{0},v_{0}\right]}\times B_{2}}^{-1}\right)=9
\]
and 
\[
h^{k}\left(K_{\mathbb{P}_{\left[u_{0},v_{0}\right]}\times B_{2}}^{-1}\right)=0
\]
for $k>0$.

ii) The quotient 
\[
\left(\mathbb{P}_{\left[u_{0},v_{0}\right]}\times B_{2}\right)^{\vee}=\frac{\mathbb{P}_{\left[u_{0},v_{0}\right]}\times B_{2}}{\left\{ C_{u,v}\right\} }
\]
carries a faithful $\mathbb{Z}_{2}$-action $T_{u,v}$.

iii) Under the branched double cover 
\[
B_{2}\rightarrow\mathbb{P}_{\left[n_{0},m_{1},m_{2}\right]}
\]
the symmetry group $S_{3}$ and associated six exceptional curves
of $D_{6}$ listed in the Table at the beginning of this Section specialize
to a subgoup $S_{3}\leq\mathrm{Aut}\left(B_{2}\right)$ and bitangents
of the branch curve $R$. In particular, $\left\{ m_{1}=0\right\} $
and $\left\{ m_{2}=0\right\} $ are bitangents to the branch locus
$R$ and satisfy
\[
\begin{array}{c}
C_{u,v}^{\ast}\left(m_{1}\right)=-m_{1}\\
C_{u,v}^{\ast}\left(m_{2}\right)=-m_{2}.
\end{array}
\]
\end{lem}

\begin{proof}
i) 
\[
K_{B_{3}}^{-1}=K_{B_{2}}^{-1}\boxtimes K_{\mathbb{P}_{\left[u_{0},v_{0}\right]}}^{-1}.
\]
Now use the Künneth formula and the Kodaira Vanishing Theorem.

ii) The $\mathbb{Z}_{4}$-action on $B_{3}$ is the action $T_{u,v}$
as constructed above. Since $C_{u,v}$ and $T_{B_{2}}$ commute and
$T_{u,v}^{2}=C_{u,v}$, the action descends to a faithful $\mathbb{Z}_{2}$-action
on the $C_{u,v}$-quotient $B_{3}^{\vee}$.

iii) Classical fact deriving from the fact that $B_{2}$ is the projection
of a cubic surface in $\mathbb{P}^{3}$ from one of its points. Under
this projection plane sections spanned by one of the $27$ lines and
the center of projection map to bitangent lines to $R\subseteq\mathbb{P}_{\left[n_{0},m_{1},m_{2}\right]}$
as does the plane section tangent to the cubic at the center of projection.
So the claim follows from the corresponding assertion for a cubic
surface with $S_{3}$ symmetry.
\end{proof}
We have the direct-sum decomposition 
\[
\left(\pi_{B_{3}^{\vee}}\right)_{\ast}K_{B_{2}\times\mathbb{P}_{\left[u_{0},v_{0}\right]}}^{-1}=K_{B_{3}^{\vee}}^{-1}\oplus\left(K_{B_{3}^{\vee}}^{-1}\otimes\mathcal{O}_{B_{3}^{\vee}}\left(\varepsilon_{e,v}\right)\right)
\]
and, referring to (\ref{eq:tee}) and (\ref{eq:cee}), we have the
following tables of eigenvectors and values for actions on anti-canonical
forms on $B_{3,\delta}$:

\smallskip{}
 
\noindent \begin{center}
\begin{tabular}{|c|c|c|}
\hline 
Table 1:  & $T_{u,v}$  & $C_{u,v}$\tabularnewline
\hline 
\hline 
$h^{0}\left(K_{B_{3,0}^{\vee}}^{-1}\right)=4$  &  & $\begin{array}{c}
C_{u,v}\left(w\right)=w\\
C_{u,v}\left(x\right)=x\\
C_{u,v}\left(y\right)=-y\\
C_{u,v}\left(z\right)=-z
\end{array}$\tabularnewline
\hline 
$\left(u_{0}+v_{0}\right)^{2}\text{·}n_{-1}$  & $-1$  & $+1$\tabularnewline
\hline 
$\left(u_{0}-v_{0}\right)^{2}\text{·}n_{-1}$  & $+1$  & $+1$\tabularnewline
\hline 
$\left(u_{0}^{2}-v_{0}^{2}\right)\text{·}m_{-i}$  & $-1$  & $+1$\tabularnewline
\hline 
$\left(u_{0}^{2}-v_{0}^{2}\right)\text{·}m_{+i}$  & $+1$  & $+1$\tabularnewline
\hline 
\end{tabular}\smallskip{}
 
\par\end{center}

\noindent \begin{center}
\begin{tabular}{|c|c|c|}
\hline 
Table 2:  & $T_{u,v}$  & $C_{u,v}$\tabularnewline
\hline 
\hline 
$h^{0}\left(K_{B_{3,0}^{\vee}}^{-1}\otimes\mathcal{O}_{B_{3}^{\vee}}\left(\varepsilon_{u,v}\right)\right)=5$  &  & $\begin{array}{c}
C_{u,v}\left(w\right)=w\\
C_{u,v}\left(x\right)=x\\
C_{u,v}\left(y\right)=-y\\
C_{u,v}\left(z\right)=-z
\end{array}$\tabularnewline
\hline 
$\left(u_{0}+v_{0}\right)^{2}\text{·}m_{-i}$  & $-i$  & $-1$\tabularnewline
\hline 
$\left(u_{0}+v_{0}\right)^{2}\text{·}m_{+i}$  & $+i$  & $-1$\tabularnewline
\hline 
$\left(u_{0}-v_{0}\right)^{2}\text{·}m_{-i}$  & $+i$  & $-1$\tabularnewline
\hline 
$\left(u_{0}-v_{0}\right)^{2}\text{·}m_{+i}$  & $-i$  & $-1$\tabularnewline
\hline 
$\left(u_{0}^{2}-v_{0}^{2}\right)\text{·}n_{-1}=:z_{0}$  & $-i$  & $-1$\tabularnewline
\hline 
\end{tabular}\smallskip{}
 
\par\end{center}
\begin{lem}
The anti-canonical linear system $\left|H^{0}\left(K_{B_{3}^{\vee}}^{-1}\right)\right|$
is numerically effective (nef) and big. It has two basepoints
\[
\left(\left\{ \left[1,\pm1\right]\right\} \times\left\{ 1,0,0\right\} \right)\in\mathbb{P}_{\left[\left(u_{0}^{2}+v_{0}^{2}\right),u_{0}v_{0}\right]}\times\mathbb{P}_{\left[n_{0},m_{1},m_{2}\right]}.
\]
Therefore $S_{\mathrm{GUT}}^{\vee}$ is a 'hyperelliptic' Enriques
surface. The rational mapping
\[
\varphi:S_{\mathrm{GUT}}^{\vee}\dashrightarrow\mathbb{P}^{3}
\]
induced by the linear system is resolved be a simple blow-up of the
two basepoints yielding a $2-1$ morphism
\[
\tilde{\varphi}:\widetilde{S_{\mathrm{GUT}}^{\vee}}\rightarrow\mathbb{P}^{3}
\]
onto a smooth quadric surface $\mathbb{P}_{\left[\left(u_{0}^{2}+v_{0}^{2}\right),u_{0}v_{0}\right]}\times\mathbb{P}_{\left[m_{1},m_{2}\right]}$
where the image the two exceptional curves is $\left\{ \left[1,\pm1\right]\right\} \times\mathbb{P}_{\left[m_{1},m_{2}\right]}$.
\end{lem}

\begin{proof}
By direct computation with Table 2 just above, $S_{\mathrm{GUT}}^{\vee}$
is identified as being in case 2c for $n=3$ in Theorem 3.2.2 of \cite{Dolgachev-1}.
Key to this classification is the fact that in Table 2 all five sections
vanish at $\left(\left\{ \left[\pm1,1\right]\right\} ,\left\{ 1,0,0\right\} \right)$.
\end{proof}

\section{Requirements for the Tate form\label{subsec:Requirements-for-the}}

For the purposes of obtaining phenomenologically consistent numerical
data,\footnote{``...phenomenologically consistent numerical data'' refers to the
necessity of having three $\mathbf{10}$ and no $\mathbf{\bar{10}}$-representations
over the matter curve $\Sigma_{\mathbf{10}}^{\left(4\right)}$, three
$\mathbf{\bar{5}}$ and no $\mathbf{5}$-representations over the
matter curve $\Sigma_{\mathbf{\bar{5}}}^{\left(41\right)}$, as well
as having exactly one Higgs doublet over the Higgs curve $\Sigma_{\mathbf{\bar{5}}}^{\left(44\right)}$.} we utilize the $B_{3}^{\wedge}$-divisors $E_{\cdot}$ in (\ref{eq:divisors})
and their corresponding rays $e_{\cdot}$ in the lattice $N^{\wedge}$.
We let the ray $f_{\cdot}$ denote the same ray in the lattice $N$
and denote the corresponding divisor as $F_{\cdot}$. Recalling that
the fan used in the construction of $B_{3}$ contains the ray generated
by $f_{\dot{x}\dot{z}\ddot{y}\ddot{w}}$ as well as its negative $-f_{\dot{x}\dot{z}\ddot{y}\ddot{w}}$
corresponding to the divisors $\left\{ u_{0}=0\right\} $ and $\left\{ v_{0}=0\right\} $
respectively.

The Tate form of our eventual $F$-theory model $W_{4}/B_{3}$ is
written as 
\[
y^{2}w=x^{3}+a_{5}xyw+a_{4}zx^{2}w+a_{3}z^{2}yw^{2}+a_{2}z^{3}xw^{2}+a_{0}z^{5}w^{3}.
\]
The spectral divisor is given by the equation
\[
0=a_{5}t^{5}+a_{4}t^{4}z+a_{3}t^{3}z{}^{2}+a_{2}t^{2}z^{3}+a_{0}z^{5}
\]
where $t=\frac{y}{x}$. By Table 2 above, $\dim H^{0}\left(K_{B_{3}}^{-1}\right)=5$
so the spaces
\[
\mathbb{C}\text{·}u_{0}v_{0}m_{1}+\mathbb{C}\text{·}\left(v_{0}^{2}+u_{0}^{2}\right)m_{2}+\mathbb{C}\text{·}z
\]
and
\[
\mathbb{C\text{·}}a_{4}+\mathbb{C\text{·}}a_{3}+\mathbb{C\text{·}}a_{2}
\]
have a one-dimensional intersection generated by
\[
\begin{array}{c}
\lambda_{1}\text{·}u_{0}v_{0}m_{1}+\lambda_{2}\text{·}\left(v_{0}^{2}+u_{0}^{2}\right)m_{2}+z\\
=\kappa_{4}\mathbb{\text{·}}a_{4}+\kappa_{3}\mathbb{\text{·}}a_{3}+\kappa_{2}\mathbb{\text{·}}a_{2}.
\end{array}
\]

\begin{lem}
In order to have the relation
\[
z=\kappa_{5}\mathbb{\text{·}}a_{5}+\kappa_{4}\mathbb{\text{·}}a_{4}+\kappa_{3}\mathbb{\text{·}}a_{3}+\kappa_{2}\mathbb{\text{·}}a_{2}
\]
required in §4 of \cite{Clemens-1}, define 
\begin{equation}
a_{5}=-\left(\lambda_{1}\text{·}u_{0}v_{0}m_{1}+\lambda_{2}\text{·}\left(v_{0}^{2}+u_{0}^{2}\right)m_{2}\right)\label{eq:a5}
\end{equation}
where $\left\{ m_{j}=0\right\} _{j=1,2}$ define the distinguished
two-orbit of $T_{0}$ acting on the $28$ bitangents to the branch
locus of $B_{2}/\mathbb{P}_{\left[n_{0},m_{1},m_{2}\right]}$.\footnote{This assumption is critical so that our model satisfies the condition
of three-generation} This pair of $C_{u,v}$- invariant bitangents $\left\{ m_{j}=0\right\} _{j=1,2}$
then lifts into the image of 
\[
F_{\dot{x}\dot{z}\ddot{y}\ddot{w}}=\left\{ \frac{e_{\dot{x}\dot{z}\ddot{y}\ddot{w}}}{2}=0\right\} \Rightarrow\left\{ u_{0}=0\right\} 
\]
and into
\[
F_{\dot{y}\dot{w}\ddot{x}\ddot{z}}=\left\{ \frac{e_{\dot{y}\dot{w}\ddot{x}\ddot{z}}}{2}=0\right\} \Rightarrow\left\{ v_{0}=0\right\} .
\]
Also 
\[
\left(v_{0}^{2}+u_{0}^{2}\right)m_{j},\,u_{0}v_{0}m_{j}\in H^{0}\left(K_{B_{3}}^{-1}\right)^{\left[-1\right]}.
\]
Therefore, for all allowable choices of $z\subseteq H^{0}\left(K_{B_{3}}^{-1}\right)^{\left[-1\right]}$
there are rational curves 
\begin{equation}
\begin{array}{c}
F_{+}\subseteq\left\{ z=0\right\} \cap F_{\dot{x}\dot{z}\ddot{y}\ddot{w}}\\
F_{-}\subseteq\left\{ z=0\right\} \cap F_{\dot{y}\dot{w}\ddot{x}\ddot{z}}
\end{array}\label{eq:spchoice1}
\end{equation}
lying in $S_{\mathrm{GUT}}$ with intersection matrix 
\[
\begin{array}{ccc}
\text{·} & F_{+} & F_{-}\\
F_{+} & -2 & 0\\
F_{-} & 0 & -2.
\end{array}
\]
$C_{u,v}$ interchanges $F_{+}$ with $F_{-}$. We choose the remaining
$a_{j}$ and $z$ and $t=\frac{y}{x}$ in the space generated by the
forms in Table 2. 

Additionally we choose $t\in H\left(K_{B_{3}}^{-1}\right)^{\left[-1\right]}$
so that it vanishes on $\left\{ z=u_{0}v_{0}=m_{2}=0\right\} $. This
choice guarantees that the surfaces 
\begin{equation}
\left(S_{\mathrm{GUT}}\cap\left(\left\{ u_{0}v_{0}=m_{2}=0\right\} \right)\right)\times\mathbb{P}_{\left[t,z\right]}\label{eq:noninherited}
\end{equation}
lie in the spectral divisor \ref{eq:specdiv-1}), in fact in the component
$\mathcal{D}^{\left(4\right)}$ of 
\begin{equation}
\mathcal{D}:=\left\{ a_{5}t^{5}+a_{4}t^{4}z+a_{3}t^{3}z^{2}+a_{2}t^{2}z^{3}+a_{0}z^{5}=0\right\} \subseteq\mathbb{P}_{\left[t,z\right]}\times B_{3}.\label{eq:spec2}
\end{equation}
\end{lem}

\begin{proof}
We need only prove the last statement. We have chosen 
\begin{equation}
t=l\left(n_{0},m_{1},m_{2}\right)\text{·}u_{0}v_{0}+q\left(u_{0},v_{0}\right)\text{·}m_{2}\label{eq:need2}
\end{equation}
namely as a generic section containing $\left\{ u_{0}v_{0}=m_{2}=0\right\} \subseteq\mathbb{P}_{\left[u_{0},v_{0}\right]}\times\mathbb{P}_{\left[n_{0},m_{1},m_{2}\right]}$
and 
\[
z=\lambda_{1}m_{1}\text{·}u_{0}v_{0}+\lambda_{2}\text{·}\left(v_{0}^{2}+u_{0}^{2}\right)m_{2}+\sum_{j=2}^{4}\kappa_{j}\text{·}a_{j}
\]
in \eqref{eq:a5}. Since $\left\{ u_{0}v_{0}=m_{2}=0\right\} \subseteq\left\{ z=0\right\} $,
$\sum_{j=2}^{4}\kappa_{j}\text{·}a_{j}$ must also vanish there. That
is, we can write
\begin{equation}
z=l'\left(n_{0},m_{1},m_{2}\right)\text{·}u_{0}v_{0}\text{·}+q'\left(u_{0},v_{0}\right)\text{·}m_{2}.\label{eq:need1}
\end{equation}
Since all sections of $H^{0}\left(\mathcal{O}_{B_{3}}\left(N\right)\right)$
are pull-backs of sections of $H^{0}\left(\mathcal{O}_{\mathbb{P}_{\left[u_{0},v_{0}\right]}\times\mathbb{P}_{\left[n_{0},m_{1},m_{2}\right]}}\left(N\right)\right)$
under the branched double cover
\[
B_{3}=\mathbb{P}_{\left[u_{0},v_{0}\right]}\times B_{2}\rightarrow\mathbb{P}_{\left[u_{0},v_{0}\right]}\times\mathbb{P}_{\left[n_{0},m_{1},m_{2}\right]},
\]
we first consider \eqref{eq:spec2} as an equation over $\mathbb{P}_{\left[u_{0},v_{0}\right]}\times\mathbb{P}_{\left[n_{0},m_{1},m_{2}\right]}$
. We consider the blow-up of \eqref{eq:spec2} in
\[
\mathbb{P}_{\left[u_{0},v_{0}\right]}\times\mathbb{P}_{\left[n_{0},m_{1},m_{2}\right]}\times\mathbb{P}_{\left[T,Z\right]}
\]
defined by 
\[
\begin{array}{c}
t=T\text{·}\xi\\
z=Z\text{·}\xi
\end{array}
\]
so that from \eqref{eq:spec2} 
\[
\left|\begin{array}{cc}
a_{4}t^{2}z^{3}+a_{2}t^{2}z^{2}+a_{0}z^{5} & 1\\
-\left(a_{5}t^{5}+a_{3}t^{3}z^{2}\right) & 1
\end{array}\right|=\xi^{5}\text{·}\left|\begin{array}{cc}
a_{4}T^{2}Z^{3}+a_{2}T^{2}Z^{2}+a_{0}Z^{5} & 1\\
-\left(a_{5}T^{5}+a_{3}T^{3}Z^{2}\right) & 1
\end{array}\right|.
\]
By \eqref{eq:need2} and \eqref{eq:need1}
\[
\left|\begin{array}{cc}
q & -l\\
u_{0}v_{0} & m_{2}
\end{array}\right|=\left|\begin{array}{cc}
q' & -l'\\
u_{0}v_{0} & m_{2}
\end{array}\right|=0
\]
and the fact that rank of the matrix
\[
\left(\begin{array}{cc}
q & -l\\
q' & -l'
\end{array}\right)
\]
is everywhere of rank at least one, the proper transform 
\begin{equation}
\begin{array}{c}
\mathcal{D}:=\left\{ \left|\begin{array}{cc}
a_{4}T^{2}Z^{3}+a_{2}T^{2}Z^{2}+a_{0}Z^{5} & 1\\
-\left(a_{5}T^{5}+a_{3}T^{3}Z^{2}\right) & 1
\end{array}\right|=0\right\} \\
\subseteq\mathbb{P}_{\left[u_{0},v_{0}\right]}\times\mathbb{P}_{\left[n_{0},m_{1},m_{2}\right]}\times\mathbb{P}_{\left[T,Z\right]}.
\end{array}\label{eq:preblow}
\end{equation}
contains the surface 
\[
\begin{array}{c}
\left\{ \xi=0\right\} \cap\mathcal{D}=\left(\left\{ u_{0}v_{0}=m_{2}=0\right\} \times\mathbb{P}_{\left[T,Z\right]}\right)\\
\subseteq\left(\mathbb{P}_{\left[u_{0},v_{0}\right]}\times\mathbb{P}_{\left[n_{0},m_{1},m_{2}\right]}\right)\times\mathbb{P}_{\left[T,Z\right]}
\end{array}
\]
projecting to
\[
\left\{ u_{0}v_{0}=m_{2}=0\right\} \subseteq\left\{ z=0\right\} \subseteq\left(\mathbb{P}_{\left[u_{0},v_{0}\right]}\times\mathbb{P}_{\left[n_{0},m_{1},m_{2}\right]}\right).
\]
With respect to the branched double cover
\[
\mathbb{P}_{\left[u_{0},v_{0}\right]}\times B_{2}\rightarrow\left(\mathbb{P}_{\left[u_{0},v_{0}\right]}\times\mathbb{P}_{\left[n_{0},m_{1},m_{2}\right]}\right)
\]
we have a cartesian diagram
\[
\begin{array}{ccc}
\mathcal{D} & \rightarrow & B_{3}\times\mathbb{P}_{\left[T,Z\right]}\\
\downarrow &  & \downarrow\\
\mathcal{D}' & \rightarrow & \left(\mathbb{P}_{\left[u_{0},v_{0}\right]}\times\mathbb{P}_{\left[n_{0},m_{1},m_{2}\right]}\right)\times\mathbb{P}_{\left[T,Z\right]}.
\end{array}
\]
The pull-back of the the exceptional set $\left\{ \xi=0\right\} \cap\mathcal{D}'\subseteq\left(\mathbb{P}_{\left[u_{0},v_{0}\right]}\times\mathbb{P}_{\left[n_{0},m_{1},m_{2}\right]}\right)\times\mathbb{P}_{\left[T,Z\right]}$
to $B_{3}\times\mathbb{P}_{\left[T,Z\right]}$ is the reducible surface
\begin{equation}
\left(F_{\pm}\cup F_{\pm}^{opp}\right)\times\mathbb{P}_{\left[T,Z\right]}\subseteq B_{3}\times\mathbb{P}_{\left[T,Z\right]}\label{eq:surf}
\end{equation}
where $F_{\pm}\cup F_{\pm}^{opp}$ denote the components of $\left\{ m_{2}=0\right\} $
in $\left\{ u_{0}=0\right\} \times B_{2}$ and $\left\{ v_{0}=0\right\} \times B_{2}$
respectively. So the spectral divisor 
\begin{equation}
\begin{array}{c}
\mathcal{D}:=\left\{ a_{5}T^{5}+a_{4}T^{4}Z+a_{3}T^{3}Z^{2}+a_{2}T^{2}Z^{3}+a_{0}Z^{5}=0\right\} \\
\subseteq B_{3}\times\mathbb{P}_{\left[T,Z\right]}
\end{array}\label{eq:sosp}
\end{equation}
contains the four-component divisor
\[
\left(F_{+}\times\mathbb{P}_{\left[T,Z\right]}\right)+\left(F_{-}\times\mathbb{P}_{\left[T,Z\right]}\right)+\left(F_{+}^{opp}\times\mathbb{P}_{\left[T,Z\right]}\right)+\left(F_{-}^{opp}\times\mathbb{P}_{\left[T,Z\right]}\right)
\]
entirely supported in $S_{\mathrm{GUT}}\times\mathbb{P}_{\left[T,Z\right]}$
. Therefore the divisor 
\begin{equation}
\left(F_{+}\times\mathbb{P}_{\left[T,Z\right]}\right)-\left(F_{-}\times\mathbb{P}_{\left[T,Z\right]}\right)\label{eq:linebun}
\end{equation}
on $\mathcal{D}$ pushes forward to the trivial divisor on $B_{3}$
but restricts to a non-trivial divisor on the threefold $S_{\mathrm{GUT}}\times\mathbb{P}_{\left[T,Z\right]}$.
Furthermore the difference $\left(F_{+}-F_{-}\right)$ defines a non-trivial
line bundle on $S_{\mathrm{GUT}}$ that has degree zero with respect
to the polarization $N$. Additionally $\beta_{3}=C_{u,s}$ exchanges
$F_{+}$ with $F_{-}$. These numerics will allow us to define a Higgs
line bundle with the correct numerical invariants. 
\end{proof}
\begin{flushleft}
To prevent the existence of vector-like exotics in our $F$-theory
model, we have required 
\par\end{flushleft}

\begin{flushleft}
\begin{equation}
a_{5}+a_{4}+a_{3}+a_{2}+a_{0}=0\label{eq:sym}
\end{equation}
so that the section of $P/B_{3}$ given by 
\begin{equation}
\begin{array}{c}
x=z^{2}w\\
y=z^{3}w
\end{array}\label{eq:sendsec}
\end{equation}
lies in $W_{4}.$ Furthermore this additional assumption forces the
spectral divisor $\mathcal{D}$ given by (\ref{eq:spec2}) to become
reducible, with one component given by $t=z$, and the other component
of degree $4$. More precisely 
\[
\begin{array}{c}
a_{5}t^{5}+a_{4}t^{4}z+a_{3}t^{3}z^{2}+a_{2}t^{2}z^{3}+a_{0}z^{5}=\\
\left(t-z\right)\left(a_{5}t^{4}+a_{54}t^{3}z-a_{20}t^{2}z^{2}-a_{0}z^{3}t-a_{0}z^{4}\right)
\end{array}
\]
where $a_{54}=a_{5}+a_{4}$ and $a_{20}=a_{2}+a_{0}$, etc. 
\par\end{flushleft}

Referring to \eqref{eq:iota} and \cite{Clemens-3}, the Higgs curve
$\Sigma_{\mathbf{\bar{5}}}^{\left(44\right)}$ is derived from the
surface in $\mathcal{D}^{\left(4\right)}$ defined by common solutions
to the $C_{u,v}$-equivariant system of equations 
\begin{equation}
\begin{array}{c}
a_{5}t^{4}-a_{20}t^{2}z^{2}-a_{0}z^{4}=0\\
a_{54}t^{2}-a_{0}z^{2}=0.
\end{array}\label{eq:Higgspush}
\end{equation}
It doubly covers the surface in $B_{3}$ defined by the resultant
equation that, using $a_{54320}=0$, reduces to 
\begin{equation}
\left|\begin{array}{cc}
a_{4} & -a_{5}\\
a_{3}+a_{0} & a_{3}
\end{array}\right|=0\label{eq:receive}
\end{equation}
with branch locus defined by the restriction of the divisor class
$N$. 

The matter curve $\Sigma_{\mathbf{\bar{5}}}^{\left(41\right)}$ is
given by the common solutions to 
\[
a_{420}=z=0.
\]
The other matter curve $\Sigma_{\mathbf{10}}^{\left(4\right)}$ is
given by 
\[
a_{5}=z=0.
\]

\subsection{Numerology of divisors on $S_{\mathrm{GUT}}$}

Furthermore (\ref{eq:justright}) implies that the genus of $Z:=N\text{·}S_{\mathrm{GUT}}$
is $7$ and 
\begin{equation}
\left\{ u_{0}=0\right\} \text{·}N^{2}=\left\{ v_{0}=0\right\} \text{·}N^{2}=2\label{eq:alsoimp}
\end{equation}
where, as above, we denote $K_{B_{3}}^{-1}=\mathcal{O}_{B_{3}}\left(N\right)$. 
\begin{prop}
\label{prop:The-line-bundle}The line bundle $\mathcal{O}_{B_{3}}\left(N\right)$
is ample and the line bundle $\mathcal{O}_{B_{3}}\left(2N\right)$
is very ample. 
\end{prop}

\begin{proof}
The line bundle $\mathcal{O}_{B_{3}}\left(N\right)$ is the pull-back
of the very ample line bundle $\mathcal{O}_{\mathbb{P}_{\left[u_{0},v_{0}\right]}}\left(2\right)\boxtimes\mathcal{O}_{\mathbb{P}_{\left[n_{0},m_{1},m_{2}\right]}}\left(1\right)$
on $\mathbb{P}_{\left[u_{0},v_{0}\right]}\times\mathbb{P}_{\left[n_{0},m_{1},m_{2}\right]}$.
By Tables 1\&2 above, the pull-back map 
\[
H^{0}\left(\mathcal{O}_{\mathbb{P}_{\left[u_{0},v_{0}\right]}\times\mathbb{P}_{\left[n_{0},m_{1},m_{2}\right]}}\left(N\right)\right)\rightarrow H^{0}\left(\mathcal{O}_{B_{3}}\left(N\right)\right)
\]
is an isomorphism. However the injective pull-back map 
\[
H^{0}\left(\mathcal{O}_{\mathbb{P}_{\left[u_{0},v_{0}\right]}\times\mathbb{P}_{\left[n_{0},m_{1},m_{2}\right]}}\left(2N\right)\right)\rightarrow H^{0}\left(\mathcal{O}_{B_{3}}\left(2N\right)\right)
\]
has a one-dimensional cokernel generated by the quartic ramification
locus of the branched double cover $B_{2}/\mathbb{P}_{\left[n_{0},m_{1},m_{2}\right]}$
and therefore contains sections that separate the two sheets. 
\end{proof}
To understand the use of these surfaces, one must consider the image
$\mathcal{O}_{B_{3}}\left(2N\right)$ is very ample. 
\[
\mathcal{C}_{Higgs}=\mathcal{C}_{Higgs}^{\left(4\right)}\cup\tilde{\tau}\left(B_{3}\right)\subseteq\tilde{W}_{4}
\]
of the spectral divisor $\mathcal{D}$ in the canonical crepant resolution
$\tilde{W}_{4}$ constructed above. $S_{\mathrm{GUT}}\subseteq B_{3}$
has canonical lifting 
\[
\tilde{S}_{\mathrm{GUT}}\subseteq\mathcal{C}_{Higgs}\subseteq\tilde{W}_{4}
\]
by means of which the Higgs line bundle is pushed forward to a line
bundle on the divisor $\mathcal{C}_{Higgs}$.

The push-forward of multiples of this line bundle on $S_{\mathrm{GUT}}\subseteq B_{3}$
will be denoted as 
\[
\mathcal{O}_{S_{\mathrm{GUT}}}\left(m\left(F_{+}-F_{-}\right)\right).
\]
The push-forward to $S_{\mathrm{GUT}}\subseteq B_{3}$ of the restriction
of the Higgs line bundle $\mathcal{L}_{Higgs}$ to $\tilde{S}_{\mathrm{GUT}}$
will then become 
\[
\mathcal{O}_{S_{\mathrm{GUT}}}\left(N+m\left(F_{+}-F_{-}\right)\right)
\]
for appropriate $m$ where $N$ is the anti-canonical divisor of $B_{3}$
(again restricted to $S_{\mathrm{GUT}}$). This twisting of the restriction
of $K_{B_{3}}^{-1}$ by this degree-zero line bundle is exactly the
modification that removes first cohomology of $\mathcal{L}_{Higgs}$
on the matter curves and, by a classical theorem in the theory of
Prym varieties, drops the number of Higgs doublets to one.

Also in our eventual $F$-theory model, $S_{\mathrm{GUT}}=\left\{ z=0\right\} $
will be linearly equivalent to $N$, matter curves will have divisor
class 
\[
Z:=S_{\mathrm{GUT}}\text{·}N
\]
and the Higgs curve $Z_{2}$ on $S_{\mathrm{GUT}}$ will have class
$2Z$.

Since $B_{2}$ is a del Pezzo surface, its anti-canonical bundle is
ample as is the anti-canonical bundle 
\[
K_{B_{2}}^{-1}\boxtimes K_{\left[u_{0},v_{0}\right]}^{-1}
\]
allowing us to apply the Kodaira Vanishing Theorem repeated in what
follows. The cohomology sequence for the short exact sequence 
\[
0\rightarrow\mathcal{O}_{B_{3}}\rightarrow\mathcal{O}_{B_{3}}\left(N\right)\rightarrow\mathcal{O}_{S_{\mathrm{GUT}}}\left(Z\right)\rightarrow0
\]
shows that the forms in the left-hand column of Tables 1 and 2 above
together form a basis for $H^{0}\left(\mathcal{O}_{B_{3}}\left(N\right)\right)$,
that 
\[
\frac{H^{0}\left(\mathcal{O}_{B_{3}}\left(N\right)\right)}{\mathbb{C}\text{·}z}\cong H^{0}\left(\mathcal{O}_{S_{\mathrm{GUT}}}\left(Z\right)\right),
\]
that $h^{1}\left(\mathcal{O}_{B_{3}}\left(N\right)\right)=h^{1}\left(\mathcal{O}_{S_{\mathrm{GUT}}}\left(Z\right)\right)$,
and that $h^{2}\left(\mathcal{O}_{B_{3}}\left(N\right)\right)=h^{2}\left(\mathcal{O}_{S_{\mathrm{GUT}}}\left(Z\right)\right)=0$.
The cohomology sequence for the short exact sequence 
\[
0\rightarrow\mathcal{O}_{S_{\mathrm{GUT}}}\left(-Z\right)\rightarrow\mathcal{O}_{S_{\mathrm{GUT}}}\rightarrow\mathcal{O}_{Z}\rightarrow0
\]
shows that $h^{1}\left(\mathcal{O}_{S_{\mathrm{GUT}}}\left(-Z\right)\right)=h^{1}\left(\mathcal{O}_{S_{\mathrm{GUT}}}\left(Z\right)\right)=0$
so that $h^{1}\left(\mathcal{O}_{B_{3}}\left(N\right)\right)=0$ as
well.

For any smooth curve $Z_{n}$ linearly equivalent to $nZ$ , the cohomology
sequence associated to the short exact sequence 
\[
0\rightarrow\mathcal{O}_{S_{\mathrm{GUT}}}\left(-\left(n-1\right)Z\right)\rightarrow\mathcal{O}_{S_{\mathrm{GUT}}}\left(Z\right)\rightarrow\mathcal{O}_{Z_{n}}\left(Z\right)\rightarrow0
\]
and Kodaira vanishing and Serre duality show that 
\begin{equation}
H^{0}\left(\mathcal{O}_{Z_{n}}\left(Z\right)\right)\cong H^{0}\left(\mathcal{O}_{S_{\mathrm{GUT}}}\left(Z\right)\right)\cong H^{2}\left(\mathcal{O}_{S_{\mathrm{GUT}}}\left(-Z\right)\right)^{\ast}\label{eq:ker2}
\end{equation}
and

\[
H^{1}\left(\mathcal{O}_{Z_{n}}\left(Z\right)\right)\cong H^{2}\left(\mathcal{O}_{S_{\mathrm{GUT}}}\left(-\left(n-1\right)Z\right)\right)\cong H^{0}\left(\mathcal{O}_{S_{\mathrm{GUT}}}\left(\left(n-1\right)Z\right)\right)^{\ast}
\]
and that all these groups have rank seven for $n=1$ and eight for
$n=2$ .

Finally, for $F_{\cdot}$ equal to $F_{+}$ or $F_{-}$, consider
\[
0\rightarrow\mathcal{O}_{B_{3}}\left(N-F_{\cdot}\right)\rightarrow\mathcal{O}_{B_{3}}\left(N\right)\rightarrow\mathcal{O}_{F_{_{\cdot}}}\left(S_{\mathrm{GUT}}\cap F_{\cdot}\right)\rightarrow0
\]
\[
0\rightarrow\mathcal{O}_{S_{\mathrm{GUT}}}\left(Z-\left(F_{\cdot}\cap S_{\mathrm{GUT}}\right)\right)\rightarrow\mathcal{O}_{S_{\mathrm{GUT}}}\left(Z\right)\rightarrow\mathcal{O}_{F_{\cdot}\cap S_{\mathrm{GUT}}}\left(Z\right)\rightarrow0
\]
and 
\[
0\rightarrow\mathcal{O}_{S_{\mathrm{GUT}}}\left(Z\right)\rightarrow\mathcal{O}_{S_{\mathrm{GUT}}}\left(Z+\left(F_{\cdot}\cap S_{\mathrm{GUT}}\right)\right)\rightarrow\mathcal{N}_{\left(F_{\cdot}\cap S_{\mathrm{GUT}}\right)|S_{\mathrm{GUT}}}\left(Z\cap F_{\cdot}\right)\rightarrow0
\]
where $\mathcal{N}_{\left(F_{\cdot}\cap S_{\mathrm{GUT}}\right)|S_{\mathrm{GUT}}}$
denotes the normal bundle of the rational curve $F_{\cdot}\cap S_{\mathrm{GUT}}$
in $S_{\mathrm{GUT}}$ so that 
\[
\mathcal{N}_{\left(F_{\cdot}\cap S_{\mathrm{GUT}}\right)|S_{\mathrm{GUT}}}\cong\mathcal{O}_{\left(F_{\cdot}\cap S_{\mathrm{GUT}}\right)}\left(-2\right).
\]

\section{$\mathbb{Z}_{2}$-quotients $B_{2}^{\vee}$ and $B_{3}^{\vee}$ and
their invariants}

Now the involution $C_{u,v}$ acts freely on $S_{\mathrm{GUT}}$ and
as 
\[
\mathcal{O}_{S_{\mathrm{GUT}}}\left(N+m\left(F_{+}-F_{-}\right)\right)\mapsto\mathcal{O}_{S_{\mathrm{GUT}}}\left(N+m\left(F_{-}-F_{+}\right)\right)
\]
Returning to 
\[
\pi_{S_{\mathrm{GUT}}^{\vee}}:S_{\mathrm{GUT}}\rightarrow S_{\mathrm{GUT}}^{\vee}:=\frac{S_{\mathrm{GUT}}}{\left\{ C_{u,v}\right\} }
\]
and letting $\mathcal{O}_{S_{\mathrm{GUT}}^{\vee}}\left(\varepsilon_{u,v}\right)$
denote the non-trivial flat (orbifold) line bundle quotient of the
trivial bundle induced by the fixpoint-free action of the involution
$C_{u,v}$, we have the splitting 
\[
\left(\pi_{S_{\mathrm{GUT}}^{\vee}}\right)_{\ast}\left(\mathcal{O}_{S_{\mathrm{GUT}}}\left(N\right)\right)=\mathcal{O}_{S_{\mathrm{GUT}}^{\vee}}\left(N\right)^{\left[+1\right]}\oplus\mathcal{O}_{S_{\mathrm{GUT}}^{\vee}}\left(N\right)^{\left[-1\right]}
\]
where 
\[
\begin{array}{c}
\mathcal{O}_{S_{\mathrm{GUT}}}\left(N\right)=\pi^{\ast}\mathcal{O}_{S_{\mathrm{GUT}}^{\vee}}\left(N\right)^{\left[+1\right]}\\
\mathcal{O}_{S_{\mathrm{GUT}}^{\vee}}\left(N\right)^{\left[-1\right]}=\mathcal{O}_{S_{\mathrm{GUT}}^{\vee}}\left(N\right)^{\left[+1\right]}\otimes\mathcal{O}_{S_{\mathrm{GUT}}^{\vee}}\left(\varepsilon_{u,v}\right).
\end{array}
\]
For any line bundle $L=\mathcal{O}_{S_{\mathrm{GUT}}}\left(D\right)$
such that
\[
C_{u,v}^{\ast}\left(D\right)=-D
\]
we have that 
\[
\pi_{\ast}\left(\mathcal{O}_{S_{\mathrm{GUT}}}\left(D\right)\right)=\mathcal{O}_{S_{\mathrm{GUT}}^{\vee}}\left(\pi_{\ast}D\right)\oplus\mathcal{O}_{S_{\mathrm{GUT}}^{\vee}}\left(\pi_{\ast}\left(-D\right)\right).
\]

\subsection{Cohomology of the Higgs bundle the Higgs curve }

We are now ready for the computation of the cohomology
\[
H^{i}\left(\Sigma_{\mathbf{\bar{5}}}^{\left(44\right)};\mathcal{L}_{\mathbf{\bar{5}}}^{\left(44\right)}\right)=H^{i}\left(\check{\Sigma}_{\mathbf{\bar{5}}}^{\left(44\right)};\mathcal{L}_{Higgs}^{\vee,\left[+1\right]}\right)\oplus H^{i}\left(\check{\Sigma}_{\mathbf{\bar{5}}}^{\left(44\right)};\mathcal{L}_{Higgs}^{\vee,\left[-1\right]}\right)
\]
of the Higgs line bundle $\mathcal{L}_{S_{\mathrm{GUT}}}=\mathcal{O}_{S_{\mathrm{GUT}}}\left(N+m\left(F_{+}-F_{-}^{opp}\right)\right)$
restricted to the Higgs curve $\Sigma_{\mathbf{\bar{5}}}^{\left(44\right)}$
as defined in \eqref{eq:receive}. We have seen in \eqref{eq:ker2}
above that on $S_{\mathrm{GUT}}$ that the restriction map
\[
H^{0}\left(\mathcal{O}_{S_{\mathrm{GUT}}}\left(N\right)\right)\rightarrow H^{0}\left(\mathcal{O}_{\Sigma_{\mathbf{\bar{5}}}^{\left(44\right)}}\left(N\right)\right)
\]
is an isomorphism. This fact has the important corollary that, by
Tables 1 \& 2 above, every section of $H^{0}\left(\mathcal{O}_{S_{\mathrm{GUT}}}\left(N\right)\right)$
is the pullback of a section of $H^{0}\left(\mathcal{O}_{\mathbb{P}_{\left[u_{0},v_{0}\right]}\times\mathbb{P}_{\left[n_{0},m_{1},m_{2}\right]}}\left(N\right)\right)$,
i.e. for the reduced image $T_{\mathrm{GUT}}\subseteq\mathbb{P}_{\left[u_{0},v_{0}\right]}\times\mathbb{P}_{\left[n_{0},m_{1},m_{2}\right]}$
of $S_{\mathrm{GUT}}\subseteq B_{3}$
\[
H^{0}\left(\mathcal{O}_{S_{\mathrm{GUT}}}\left(N\right)\right)=\rho^{\ast}H^{0}\left(\mathcal{O}_{T_{\mathrm{GUT}}}\left(N\right)\right).
\]
Therefore for $m>0$ we have the commutative diagram
\[
\begin{array}{ccc}
H^{0}\left(\mathcal{O}_{T_{\mathrm{GUT}}}\left(N-m\cdot\left\{ m_{2}=0\right\} \right)\right) & \hookrightarrow & H^{0}\left(\mathcal{O}_{T_{\mathrm{GUT}}}\left(N\right)\right)\\
\updownarrow= &  & \updownarrow=\\
H^{0}\left(\mathcal{O}_{S_{\mathrm{GUT}}}\left(N-m\cdot\left(F_{\pm}+F_{\pm}^{opp}\right)\right)\right) & \hookrightarrow & H^{0}\left(\mathcal{O}_{S_{\mathrm{GUT}}}\left(N\right)\right)\\
\updownarrow= &  & \updownarrow=\\
H^{0}\left(\mathcal{O}_{S_{\mathrm{GUT}}}\left(N-m\cdot F_{\pm}\right)\right) & \hookrightarrow & H^{0}\left(\mathcal{O}_{S_{\mathrm{GUT}}}\left(N\right)\right)\\
\updownarrow= &  & \updownarrow=\\
H^{0}\left(\mathcal{O}_{\Sigma_{\mathbf{\bar{5}}}^{\left(44\right)}}\left(\Sigma_{\mathbf{\bar{5}}}^{\left(44\right)}\cdot\left(N-m\cdot F_{\pm}\right)\right)\right) & \hookrightarrow & H^{0}\left(\mathcal{O}_{\Sigma_{\mathbf{\bar{5}}}^{\left(44\right)}}\left(\Sigma_{\mathbf{\bar{5}}}^{\left(44\right)}\cdot N\right)\right)
\end{array}
\]
where all the vertical maps are isomorphisms. Since $\mathcal{O}_{\Sigma_{\mathbf{\bar{5}}}^{\left(44\right)}}\left(\Sigma_{\mathbf{\bar{5}}}^{\left(44\right)}\cdot N\right)$
is a theta characteristic and 
\[
C_{u,v}^{\ast}\left(\mathcal{O}_{\Sigma_{\mathbf{\bar{5}}}^{\left(44\right)}}\left(\Sigma_{\mathbf{\bar{5}}}^{\left(44\right)}\cdot\left(m\cdot\left(F_{+}-F_{-}\right)\right)\right)\right)=\mathcal{O}_{\Sigma_{\mathbf{\bar{5}}}^{\left(44\right)}}\left(\Sigma_{\mathbf{\bar{5}}}^{\left(44\right)}\cdot\left(m\cdot\left(F_{-}-F_{+}\right)\right)\right)
\]
we are exactly in the situation of Step II of Lemma $1$ in \cite{Mumford}.
When $m=2$ the spaces on the left in the inclusion diagram are zero
since no non-zero section of $H^{0}\left(\mathcal{O}_{T_{\mathrm{GUT}}}\left(N\right)\right)$
vanishes to second-order on $\left\{ m_{2}=0\right\} $. Therefore
the eight points of $2\left(\Sigma_{\mathbf{\bar{5}}}^{\left(44\right)}\cdot\left(F_{-}+F_{-}^{opp}\right)\right)$
impose the independent conditions required in Step II for the proof
of the Lemma. We start from $m=1$ where

\[
\begin{array}{c}
H^{0}\left(\mathcal{O}_{\Sigma_{\mathbf{\bar{5}}}^{\left(44\right)}}\left(\Sigma_{\mathbf{\bar{5}}}^{\left(44\right)}\cdot\left(N+\left(F_{+}-F_{-}\right)\right)\right)\right)=\\
H^{0}\left(\mathcal{O}_{\Sigma_{\mathbf{\bar{5}}}^{\left(44\right)}}\left(\Sigma_{\mathbf{\bar{5}}}^{\left(44\right)}\cdot\left(N-F_{-}\right)\right)\right)=\\
H^{0}\left(\mathcal{O}_{S_{\mathrm{GUT}}}\left(N-F_{-}\right)\right)=H^{0}\left(\mathcal{O}_{\mathbb{P}_{\left[u_{0},v_{0}\right]}}\left(2\right)\right)\cdot m_{2}.
\end{array}
\]
Since $F_{-}\cdot\Sigma_{\mathbf{\bar{5}}}^{\left(44\right)}=2$ the
two points impose independent conditions on $H^{0}\left(\mathcal{O}_{S_{\mathrm{GUT}}}\left(N-F_{-}\right)\right)$
so we conclude by Step II of Lemma $1$ in \cite{Mumford} that 
\begin{equation}
1=h^{0}\left(\mathcal{O}_{\Sigma_{\mathbf{\bar{5}}}^{\left(44\right)}}\left(\Sigma_{\mathbf{\bar{5}}}^{\left(44\right)}\left(N-2F_{-}\right)\right)\right)=h^{0}\left(\mathcal{O}_{\Sigma_{\mathbf{\bar{5}}}^{\left(44\right)}}\left(\Sigma_{\mathbf{\bar{5}}}^{\left(44\right)}\cdot\left(N+2\left(F_{+}-F_{-}\right)\right)\right)\right)\label{eq:commdiag}
\end{equation}
confirmed by the fact that only one section of $H^{0}\left(\mathcal{O}_{S_{\mathrm{GUT}}}\left(N-F_{-}\right)\right)$,
namely $v_{0}^{2}\cdot m_{2}$ vanishes additionally to order $2$
at the points of $\left(\Sigma_{\mathbf{\bar{5}}}^{\left(44\right)}\cap F_{-}\right)$
and that the single generator of $H^{0}\left(\mathcal{O}_{\Sigma_{\mathbf{\bar{5}}}^{\left(44\right)}}\left(\Sigma_{\mathbf{\bar{5}}}^{\left(44\right)}\cdot\left(N+2\left(F_{+}-F_{-}\right)\right)\right)\right)$
is given by 
\begin{equation}
\frac{u_{0}^{2}}{v_{0}^{2}}m_{2}\in H^{0}\left(\Sigma_{\mathbf{\bar{5}}}^{\left(44\right)};\mathcal{L}_{\mathbf{\bar{5}}}^{\left(44\right)}\right)\label{eq:unique}
\end{equation}
 and the single generator of $H^{0}\left(\mathcal{O}_{\Sigma_{\mathbf{\bar{5}}}^{\left(44\right)}}\left(\Sigma_{\mathbf{\bar{5}}}^{\left(44\right)}\cdot\left(N+2\left(F_{-}-F_{+}\right)\right)\right)\right)$
is given by
\[
\frac{v_{0}^{2}}{u_{0}^{2}}m_{2}\in H^{0}\left(\Sigma_{\mathbf{\bar{5}}}^{\left(44\right)};C_{u,v}^{\ast}\left(\mathcal{L}_{\mathbf{\bar{5}}}^{\left(44\right)}\right)\right)
\]

To compute the symmetric/anti-symmetric decomposition of this generator
on $\check{\Sigma}_{\mathbf{\bar{5}}}^{\left(44\right)}$ in
\[
\mathcal{L}_{Higgs}^{\vee,\left[+1\right]}\oplus C_{u,v}^{\ast}\left(\mathcal{L}_{Higgs}^{\vee,\left[+1\right]}\right)=\mathcal{L}_{Higgs}^{\vee,\left[+1\right]}\oplus\mathcal{L}_{Higgs}^{\vee,\left[-1\right]}
\]
we consider local sections $\left(\gamma_{1}\oplus\delta_{1}\right)$
and $\left(\gamma_{2}\oplus\delta_{2}\right)$ of
\[
\mathcal{L}_{Higgs}\oplus C_{u,v}^{\ast}\left(\mathcal{L}_{Higgs}\right)
\]
at points $p_{1}$ and $p_{2}=C_{u,v}\left(p_{1}\right)$ and write
the expression
\[
\left(\gamma_{1}-C_{u,v}^{\ast}\left(\delta_{2}\right),\delta_{1}-C_{u,v}^{\ast}\left(\gamma_{2}\right)\right).
\]
If this expression is $\left(0,0\right)$ we have a symmetric section,
that is, the pull-back of a local section of $\mathcal{L}_{Higgs}^{\vee,\left[+1\right]}$
at $p\in\check{\Sigma}_{\mathbf{\bar{5}}}^{\left(44\right)}$ with
inverse image $\left\{ p_{1},p_{2}\right\} $. On the other hand,
if 
\[
\left(\gamma_{1}+C_{u,v}^{\ast}\left(\delta_{2}\right),\delta_{1}+C_{u,v}^{\ast}\left(\gamma_{2}\right)\right)=\left(0,0\right)
\]
the section is antisymmetric. 

In our case
\[
\begin{array}{cc}
\gamma_{1}=v_{0}^{2}m_{2} & \gamma_{2}=u_{0}^{2}m_{2}\\
\delta_{1}=u_{0}^{2}m_{2} & \delta_{2}=v_{0}^{2}m_{2}.
\end{array}
\]
Since $C_{u,v}^{\ast}\left(m_{2}\right)=-m_{2}$ the symmetric summand
\[
\begin{array}{c}
\gamma_{1}+C_{u,v}^{\ast}\left(\delta_{2}\right)=v_{0}^{2}m_{2}+C_{u,v}^{\ast}\left(v_{0}^{2}m_{2}\right)\\
=v_{0}^{2}m_{2}+\left(-v_{0}^{2}m_{2}\right)=0
\end{array}
\]
is zero and the unique non-zero section 
\[
\begin{array}{c}
\gamma_{1}-C_{u,v}^{\ast}\left(\delta_{2}\right)=v_{0}^{2}m_{2}-C_{u,v}^{\ast}\left(v_{0}^{2}m_{2}\right)\\
=v_{0}^{2}m_{2}-\left(-v_{0}^{2}m_{2}\right)=2v_{0}^{2}m_{2},
\end{array}
\]
the image of $2m_{2}v_{0}^{2}\in\rho^{\ast}H^{0}\left(\mathcal{O}_{T_{\mathrm{GUT}}}\left(N\right)\right)$,
generates the anti-symmetric summand. 
\begin{lem}
\label{lem:The-normal-functions}We select $m=2$ for our final refinement
of the definition of the Higgs line bundle 
\[
\mathcal{L}_{Higgs}=\mathcal{O}_{\mathcal{D}_{5}}\left(N+2\left(F_{+}-F_{-}\right)\right)
\]
on the spectral variety $\mathcal{D}_{5}=\mathcal{D}_{4}\cup\mathcal{D}_{1}$.
.There is a unique non-zero section \eqref{eq:unique} of the Higgs
line bundle on the Higgs curve $\Sigma_{\mathbf{\bar{5}}}^{\left(44\right)}$.
Its image under the direct sum decomposition 
\begin{equation}
H^{i}\left(\Sigma_{\mathbf{\bar{5}}}^{\left(44\right)};\mathcal{L}_{\mathbf{\bar{5}}}^{\left(44\right)}\right)=H^{i}\left(\check{\Sigma}_{\mathbf{\bar{5}}}^{\left(44\right)};\mathcal{L}_{Higgs}^{\vee,\left[+1\right]}\right)\oplus H^{i}\left(\check{\Sigma}_{\mathbf{\bar{5}}}^{\left(44\right)};\mathcal{L}_{Higgs}^{\vee,\left[-1\right]}\right)\label{eq:kernel}
\end{equation}
lies in the anti-symmeric summand. 
\end{lem}

\section{Asymptotic $\mathbb{Z}_{4}$ $\mathbf{R}$-symmetry \label{subsec:Modification-for-}}

\subsection{Asymptotic Tate form }

We next examine an 'asymptotic' $\mathbb{Z}_{4}$ R-symmetry for these
constructions. For this we will have to lift the $\mathbb{Z}_{4}$-action
on $B_{3}$ described in Tables 1\&2 above by lifting it to a $\mathbb{Z}_{4}$-action
on the semi-stable degeneration $W_{4,0}$ introduced in \eqref{eq:semi}.
The lifting is described by

\smallskip{}
 
\noindent \begin{center}
\begin{tabular}{|c|c|c|}
\hline 
Table 1$^{asymp}$:  & $T_{u,v}$  & $C_{u,v}$\tabularnewline
\hline 
\hline 
$h^{0}\left(K_{B_{3,0}^{\vee}}^{-1}\right)=4$  & $\begin{array}{c}
T_{u,v}\left(w\right)=w\\
T_{u,v}\left(x\right)=-x\\
T_{u,v}\left(y\right)=iy
\end{array}$  & $\begin{array}{c}
C_{u,v}\left(w\right)=w\\
C_{u,v}\left(x\right)=x\\
C_{u,v}\left(y\right)=-y\\
C_{u,v}\left(z\right)=-z
\end{array}$\tabularnewline
\hline 
$\left(u_{0}+v_{0}\right)^{2}\text{·}n_{-1}$  & $-1$  & $+1$\tabularnewline
\hline 
$\left(u_{0}-v_{0}\right)^{2}\text{·}n_{-1}$  & $+1$  & $+1$\tabularnewline
\hline 
$\left(u_{0}^{2}-v_{0}^{2}\right)\text{·}m_{-i}$  & $-1$  & $+1$\tabularnewline
\hline 
$\left(u_{0}^{2}-v_{0}^{2}\right)\text{·}m_{+i}$  & $+1$  & $+1$\tabularnewline
\hline 
\end{tabular}\smallskip{}
 
\par\end{center}

\noindent \begin{center}
\begin{tabular}{|c|c|c|}
\hline 
Table 2$^{asymp}$: & $T_{u,v}$  & $C_{u,v}$\tabularnewline
\hline 
\hline 
$h^{0}\left(K_{B_{3,0}^{\vee}}^{-1}\otimes\mathcal{O}_{B_{3,0}^{\vee}}\left(\varepsilon_{u,v}\right)\right)=5$  & $\begin{array}{c}
T_{u,v}\left(w\right)=w\\
T_{u,v}\left(x\right)=-x\\
T_{u,v}\left(y\right)=iy
\end{array}$  & $\begin{array}{c}
C_{u,v}\left(w\right)=w\\
C_{u,v}\left(x\right)=x\\
C_{u,v}\left(y\right)=-y\\
C_{u,v}\left(z\right)=-z
\end{array}$\tabularnewline
\hline 
$\left(u_{0}+v_{0}\right)^{2}\text{·}m_{-i}$  & $-i$  & $-1$\tabularnewline
\hline 
$\left(u_{0}+v_{0}\right)^{2}\text{·}m_{+i}$  & $+i$  & $-1$\tabularnewline
\hline 
$\left(u_{0}-v_{0}\right)^{2}\text{·}m_{-i}$  & $+i$  & $-1$\tabularnewline
\hline 
$\left(u_{0}-v_{0}\right)^{2}\text{·}m_{+i}$  & $-i$  & $-1$\tabularnewline
\hline 
$\left(u_{0}^{2}-v_{0}^{2}\right)\text{·}n_{-1}=:z_{0}$  & $-i$  & $-1$\tabularnewline
\hline 
\end{tabular}\smallskip{}
\par\end{center}

In the $F$-theory models we are proposing, the semi-stable degeneration
\[
W_{4,\delta}\Rightarrow W_{4,0}=dP_{a}\cup dP_{b}
\]
to the union of two bundles of del Pezzo surfaces over $B_{2}$ we
must introduce deformations 
\[
\begin{array}{c}
a_{j,\delta}=\delta a_{j}+\left(1-\delta\right)a_{j,0}\\
z_{\delta}=\delta z+\left(1-\delta\right)z_{0}\\
t_{\delta}=\delta t+\left(1-\delta\right)t_{0}
\end{array}
\]
of the Tate form (\ref{eq:Tate}) such that $a_{j,0},z_{0},t_{0}$
all are constrained to lie in the three-dimensional $\left(-i\right)$-eigenspace
for $T_{u,v}$ spanned by the $\left(-i\right)$-eigenvectors for
$T_{u,v}$ given in Table 2$^{asymp}$ and so that the action of $T_{u,v}$
on $\left[w,x,y\right]$ is also as indicated in Table 2$^{asymp}$.
These conditions will insure the existence of an asymptotic $\mathbb{Z}_{4}$
$\mathbf{R}$-symmetry as we discuss next.

First of all
\begin{equation}
\left[u_{0},v_{0}\right]\mapsto\left[u_{0},v_{0}\right]\left[\begin{array}{cc}
\frac{1+i}{2} & \frac{1-i}{2}\\
\frac{1-i}{2} & \frac{1+i}{2}
\end{array}\right]\label{eq:act}
\end{equation}
that is
\[
T_{u,v}^{\ast}\left(\left[\begin{array}{c}
u_{0}\\
v_{0}
\end{array}\right]\right)=\left[\begin{array}{c}
u_{0}\circ T_{u,v}\\
v_{0}\circ T_{u,v}
\end{array}\right]=\left[\begin{array}{cc}
\frac{1+i}{2} & \frac{1-i}{2}\\
\frac{1-i}{2} & \frac{1+i}{2}
\end{array}\right]\text{·}\left[\begin{array}{c}
u_{0}\\
v_{0}
\end{array}\right]
\]
so that $T_{u,v}^{\ast}\left(u_{0}\text{·}v_{0}\right)=\frac{1}{2}\left(u_{0}^{2}+v_{0}^{2}\right)$
and $\left(T_{u,v}^{\ast}\right)\left(\frac{u_{0}^{2}+v_{0}^{2}}{2}\right)=u_{0}v_{0}$.
Next from \eqref{eq:t0} $T_{u,v}^{\ast}\left(m_{1}\right)=m_{2}$
and $\left(T_{u,v}^{\ast}\right)\left(m_{2}\right)=-m_{1}$.

We will require that 
\[
z_{0}=\left(u_{0}^{2}-v_{0}^{2}\right)\text{·}n_{-1}
\]
so that 
\[
\begin{array}{c}
T_{u,v}^{\ast}\left(z_{0}\right)=-\left(\left(\frac{1+i}{2}u_{0}+\frac{1-i}{2}v_{0}\right)^{2}-\left(\frac{1-i}{2}u_{0}+\frac{1+i}{2}v_{0}\right)^{2}\right)\text{·}n_{-1}\\
-\left(\frac{i}{2}\left(u_{0}^{2}-v_{0}^{2}\right)-\frac{i}{2}\left(-\left(u_{0}^{2}-v_{0}^{2}\right)\right)\right)\text{·}n_{-1}=\\
-i\text{·}z_{0}
\end{array}
\]
As well we require that
\[
a_{5,0}=u_{0}v_{0}\text{·}m_{1}+i\text{·}\left(\frac{u_{0}^{2}+v_{0}^{2}}{2}\right)m_{2}
\]
so that
\[
\begin{array}{c}
\left(T_{u,v}^{\ast}\right)\left(a_{5,0}\right)\\
=\left(\frac{u_{0}^{2}+v_{0}^{2}}{2}\right)m_{2}-i\text{·}u_{0}v_{0}\text{·}m_{1}\\
=-i\text{·}\left(u_{0}v_{0}\text{·}m_{1}+i\text{·}\left(\frac{u_{0}^{2}+v_{0}^{2}}{2}\right)m_{2}\right)\\
=-i\text{·}a_{5.0}.
\end{array}
\]
Then, referring to (\ref{eq:a5}) 
\[
a_{5,\delta}\in\mathbb{C}\text{·}\left(u_{0}v_{0}\text{·}m_{1}+i\text{·}\left(\frac{u_{0}+v_{0}}{2}\right)m_{2}\right)+\delta\text{·}\left(\mathbb{C}\left(v_{0}^{2}+u_{0}^{2}\right)+\mathbb{C}u_{0}v_{0}\right)m_{2}+\mathbb{C}\text{·}z_{\delta}
\]
and the necessary twisting divisor $F_{+}-F_{-}$ lie in $S_{\mathrm{GUT}}$
for all $\delta$.

To achieve a $\mathbb{Z}_{4}$ $\mathbf{R}$-symmetry in the limiting
$W_{4,0}$ we first note that in the restriction Tate form to the
Heterotic model, the sections $a_{j},\,z,\,\frac{y}{x}\in H^{0}\left(K_{B_{3}}^{-1}\right)$
specialize by definition to sections 
\[
a_{j,0},\,z_{0},\,\frac{y_{0}}{x_{0}}\in H^{0}\left(K_{B_{2}}^{-1}\right)=\mathbb{C}\text{·}m_{+i}+\mathbb{C}\text{·}m_{-i}+\mathbb{C}\text{·}n_{-1}
\]
whose variation in the fiber variable $a=\frac{u_{0}-v_{0}}{u_{0}+v_{0}}$
of $B_{3}/B_{2}$ is absorbed in the $dP_{9}$ fiber while $\left[n_{-1},m_{+i},m_{-i}\right]$
in \eqref{eq:t1} are the coordinates of the base of the Heterotic
model $\tilde{V}_{3}/B_{2}$. Secondly the action of $T_{u,v}$ on
$H^{0}\left(\mathcal{O}_{\mathbb{P}_{\left[u_{0},v_{0}\right]}}\left(2\right)\right)$,
that is, on the fibers of $B_{3,0}/B_{2}$ has eigenvalues $+i$,
$+1$, and $-1$. Then , so 
\[
\begin{array}{c}
q_{+1}\left(u_{0},v_{0}\right)=u_{0}v_{0}+\frac{1}{2}\left(u_{0}^{2}+v_{0}^{2}\right)\\
q_{-1}\left(u_{0},v_{0}\right)=u_{0}v_{0}-\frac{1}{2}\left(u_{0}^{2}+v_{0}^{2}\right)\\
q_{+i}\left(u_{0},v_{0}\right)=u_{0}^{2}-v_{0}^{2}.
\end{array}
\]
Then the $-i$-eigenspace in Table 2 is spanned by the three vectors
\[
q_{+i}\text{·}n_{-1},\,q_{+1}\text{·}m_{-i},\,q_{-1}\text{·}m_{+i}.
\]
$\left(T_{u,v}^{\ast}\right)^{2}\left(u_{0}\text{·}v_{0}\right)=C_{u,v}^{\ast}\left(u_{0}\text{·}v_{0}\right)=v_{0}\text{·}u_{0}$
as required for compatibility with the definition of the Higgs bundle
for all $\delta$.

Proceeding in this way we can arrange so that the action of $T_{u,v}$
on the Heterotic model is given by the action of 
\[
\begin{array}{c}
a_{j,0}\circ T_{u,v}=\left(-i\right)\text{·}a_{j,0}\,\,j=2,3,4,5\\
a_{0,0}=-\left(a_{2,0}+a_{3,0}+a_{4,0}+a_{5,0}\right).\\
z_{0}=\kappa_{2,0}\text{·}a_{2,0}+\kappa_{3,0}\text{·}a_{3,0}+\kappa_{4,0}\text{·}a_{4,0}+\kappa_{5,0}\text{·}a_{5,0}.
\end{array}
\]
Notice again that these definitions imply that $T_{u,v}$ acts on
the Heterotic model as $T_{u,v}$ on $V_{3}/B_{2}$ and as (\ref{eq:act})
on the $dP_{9}$-bundles. 
\begin{lem}
\label{lem:The-Tate-form}Under the above assumptions, the specialization
of the Tate form and $B_{2,\delta}$ (\ref{eq:Tate}) to $\delta=0$
is taken to minus itself under the action of $T_{u,v}$. We will therefore
say that our $F$-theory model satisfies an asymptotic $\mathbb{Z}_{4}$
\textbf{R}-symmetry. 
\end{lem}

\begin{proof}
One checks directly using (\ref{eq:Tate}) and the eigenvalues in
the $T_{u,v}$-column in the above Table 2$^{asymp}$ that the Tate
form is taken to minus itself. 
\end{proof}
\begin{flushleft}
Therefore $T_{u,v}$ takes the holomorphic four-form on the semi-stable
limit $W_{4,0}$ to minus itself \cite{Davies}. As a consequence
$T_{u,v}$ acts trivially on the global section of $K_{S_{\mathrm{GUT}}}$.
Furthermore $T_{u,v}$ commutes with $C_{u,v}$ since the action of
$T_{u,v}^{2}$ coincides with the action of $C_{u,v}$. Thus $T_{u,v}$
will be our candidate for the asymptotic $\mathbb{Z}_{4}$ $\mathbf{R}$-symmetry
on the quotient Calabi-Yau fourfold 
\[
W_{4,0}^{\vee}/B_{3,0}^{\vee}:=\frac{W_{4,0}/B_{3,0}}{\left\{ C_{u,v}\right\} }.
\]
\par\end{flushleft}

\subsection{Charges for $\mathbb{Z}_{4}$-\textmd{R} symmetry \cite{Lee}}

The generator $T_{u,v}$ of our asymptotic $\mathbb{Z}_{4}$ $\mathrm{\mathbf{R}}$-symmetry
on $W_{4,0}$ and $W_{4,0}^{\vee}$ has the defining equations of
$\Sigma_{\mathbf{10}}^{\left(4\right)}$ and $\Sigma_{\mathbf{\bar{5}}}^{\left(41\right)}$
as $-i$-eigenvectors and the defining equation of $\Sigma_{\mathbf{\bar{5}}}^{\left(44\right)}$
as $\left(-1\right)$-eigenvector. Therefore, since $z_{0}$ has eigenvalue
$\left(-i\right)$ , we apply these values in \cite{Clemens-3} to
obtain the following table: 
\begin{center}
\smallskip{}
\par\end{center}

\begin{center}
\begin{tabular}{|c|c|c|c|}
\hline 
TABLE 3:\quad{}$T_{u,v}$  & Equation  & $T_{u,v}$-charge  & states\tabularnewline
\hline 
\hline 
matter fields on $\frac{\Sigma_{\mathbf{10}}^{\left(4\right)}}{\left\{ C_{u,v}\right\} }$  & $a_{5}=z=0$  & $-1$  & $H^{0}\left(\frac{\Sigma_{\mathbf{10}}^{\left(4\right)}}{\left\{ C_{u,v}\right\} };\mathcal{L}_{Higgs}^{\vee,\left[\pm1\right]}\right)$\tabularnewline
\hline 
matter fields on $\frac{\Sigma_{\mathbf{\bar{5}}}^{\left(41\right)}}{\left\{ C_{u,v}\right\} }$  & $a_{420}=z=0$  & $-1$  & $H^{0}\left(\frac{\Sigma_{\mathbf{\bar{5}}}^{\left(41\right)}}{\left\{ C_{u,v}\right\} };\mathcal{L}_{Higgs}^{\vee,\left[\pm1\right]}\right)$\tabularnewline
\hline 
Higgs fields on $\frac{\Sigma_{\mathbf{\bar{5}}}^{\left(44\right)}}{\left\{ C_{u,v}\right\} }$  & $\begin{array}{c}
\left|\begin{array}{cc}
a_{4} & -a_{5}\\
a_{3}+a_{0} & a_{3}
\end{array}\right|\\
=z=0
\end{array}$  & $+i$  & $\begin{array}{c}
H^{0}\left(\frac{\Sigma_{\mathbf{\bar{5}}}^{\left(44\right)}}{\left\{ C_{u,v}\right\} };\mathcal{L}_{Higgs}^{\vee,\left[-1\right]}\right)\\
H^{1}\left(\frac{\Sigma_{\mathbf{\bar{5}}}^{\left(44\right)}}{\left\{ C_{u,v}\right\} };\mathcal{L}_{Higgs}^{\vee,\left[-1\right]}\right)
\end{array}$\tabularnewline
\hline 
bulk matter on $\frac{S_{\mathrm{GUT}}}{\left\{ C_{u,v}\right\} }$  & $z=0$  & $-i$  & $H^{2}\left(K_{\frac{S_{\mathrm{GUT}}}{\left\{ C_{u,v}\right\} }}\right)$\tabularnewline
\hline 
\end{tabular}
\par\end{center}

\section{Conclusion}

Rather than inventing a base-space $B_{3}$ and fine-tuning it to
yield the right invariants for a phenomenologically consistent $F$-theory,
we have adopted the philosophy that the representation theory required
by the physics will dictate the base space for the Tate form and ensuing
$F$-theory. Perhaps surprisingly, the representation theory contains
almost completely within itself one and only one phenomenologically
consistent $F$-theory. The detailed presentation of that model, based
on the construction and analysis of the $B_{3}$ presented in this
paper, is the subject of the companion paper \cite{Clemens-3}.

\section*{Acknowledgments}

The authors would like to thank Dave Morrison, Tony Pantev and Sakura
Schäfer-Nameki for their guidance and many helpful conversations over
several years. However the authors themselves take sole responsibility
for any errors or omissions in this series of papers. S.R. acknowledges
partial support from Department of Energy grant DE-SC0011726.

\end{document}